\documentclass[10pt, one column, draft cls]{IEEEtran}				% decomment for supplementary material
\usepackage{amsthm}     % math
\IEEEoverridecommandlockouts

\usepackage[english]{babel}                 % needed
\usepackage{amsmath,amssymb,amsfonts,hyperref} 
\usepackage{bbold}    % math
\usepackage{dsfont}
\usepackage{color,bbm}                          % \color{}
\usepackage[normalem]{ulem}                           % strike on text \sout{}; without [normalem] emph is replaced by strike below text
\usepackage{tikz}
\usepackage{soul}
\usepackage[linesnumbered,ruled,vlined]{algorithm2e}
\usepackage{algcompatible}

\usepackage{subfigure}

\def\beq{\begin{equation}}
\def\eeq{\end{equation}}

\newcommand{\ba}{\begin{array}}
\newcommand{\ea}{\end{array}}

\renewcommand{\l}{\left}\renewcommand{\r}{\right}
\newcommand{\be}{\begin{equation}}
\newcommand{\ee}{\end{equation}}

\newcommand{\R}{\mathbb{R}}
\newcommand{\N}{\mathbb{N}}

\newcommand{\U}{\mathcal{U}}
\newcommand{\V}{\mathcal{V}}

\newcommand{\de}{\mathrm{d}}

\newcommand{\se}{\text{ if }}

\DeclareMathOperator*{\argmax}{argmax}
\DeclareMathOperator*{\argmin}{argmin}

\renewcommand{\natural}{{\mathbb{N}}}

\newcommand{\real}{{\mathbb{R}}}
\newcommand{\realnoneg}{\R_+}

\newcommand{\rlb}{r}

\newcommand{\setdef}[2]{\{#1 \; | \; #2\}}
\newcommand{\until}[1]{\{1,\dots,#1\}}
\newcommand{\zerountil}[1]{\{0,\dots,#1\}}

\newcommand{\mc}{\mathcal}
\def\N{\mathbb{N}}
\newcommand{\flowmax}{C}
\newcommand{\BPAflowmax}{\Gamma}

\newcommand{\Dcum}{\triangle}
\newcommand{\Dnorm}{\mc D}

\newcommand{\onebf}{\mathbf{1}}
\newcommand{\zerobf}{\mathbf{0}}

\newcommand{\mor}{\mc R}

\newcommand{\minCnode}{\underbar{\flowmax}}
\newcommand{\maxCnode}{\bar{\flowmax}}

\newcommand{\tdelta}{\tilde{\delta}}

\newtheorem{assumption}{Assumption}

\newtheorem{theorem}{Theorem}
\newtheorem{proposition}{Proposition}

\newtheorem{definition}{Definition}
\newtheorem{lemma}{Lemma}
\newtheorem{remark}{Remark}

\newtheorem{example}{Example}

\newcommand\oprocendsymbol{\hbox{$\square$}}
\newcommand\oprocend{\relax\ifmmode\else\unskip\hfill\fi\oprocendsymbol}

\newcommand{\conciseversiononly}[1]{}
\newcommand{\memoryversion}[1]{}

\begin{document}

\title{Robust Network Routing under Cascading Failures}

%\author{Ketan Savla \quad Giacomo Como \quad Munther A. Dahleh}

\author{Ketan Savla \quad Giacomo Como \quad Munther A. Dahleh
\thanks{K. Savla is with the Sonny Astani Department of Civil and Environmental Engineering at the University of Southern California, Los Angeles, CA, USA. \texttt{ksavla@usc.edu}. G. Como is with the Department of Automatic Control, Lund University, Lund, Sweden. \texttt{giacomo@control.lth.edu}. M. A. Dahleh is with the Laboratory for Information and Decision Systems at the Massachusetts Institute of Technology, Cambridge, MA, USA. 
\texttt{dahleh@mit.edu}}
%\thanks{G. Como is with the Department of Automatic Control, Lund University, Lund, Sweden. \texttt{giacomo@control.lth.edu}}
%\thanks{M. A. Dahleh is with the Laboratory for Information and Decision Systems at the Massachusetts Institute of Technology, Cambridge, MA, USA.
%\texttt{dahleh@mit.edu}.}
} 

\maketitle

\begin{abstract}
We propose a dynamical model for cascading failures in single-commodity network flows. In the proposed model, the network state consists of flows \emph{and} activation status of the links. Network dynamics is determined by a, possibly state-dependent and adversarial, disturbance process that reduces flow capacity on the links, and routing policies at the nodes that have access to the network state, but are oblivious to the presence of disturbance. Under the proposed dynamics, a link becomes irreversibly inactive either due to overload condition on itself or on all of its immediate downstream links. The coupling between link activation and flow dynamics implies that links to become inactive successively are not necessarily adjacent to each other, and hence the pattern of cascading failure under our model is qualitatively different than standard cascade models. The magnitude of a disturbance process is defined as the sum of cumulative capacity reductions across time and links of the network, and the margin of resilience of the network is defined as the infimum over the magnitude of all disturbance processes under which the links at the origin node become inactive. We propose an algorithm to compute an upper bound on the margin of resilience for the setting where the routing policy only has access to information about the local state of the network. For the limiting case when the routing policies update their action as fast as network dynamics, we identify sufficient conditions on network parameters under which the upper bound is tight under an appropriate routing policy. Our analysis relies on making connections between network parameters and monotonicity in network state evolution under proposed dynamics.   
%We propose a dynamical model for cascading failures in network flows, where the network state consists of flows and activation status of the links. Network dynamics is determined by a disturbance process that reduces flow capacity on links, and routing policies at the nodes that have access to the network state, but are oblivious to the presence of disturbance. Under the proposed dynamics, a link becomes irreversibly inactive either due to overload condition on itself or on all of its immediate downstream links. The coupling between link activation and flow dynamics implies that links to become inactive successively are not necessarily adjacent to each other. The magnitude of a disturbance process is defined as the cumulative capacity reduction across time and links of the network, and the margin of resilience of the network as the infimum over the magnitude of all disturbance processes under which the links at the origin node become inactive. We propose an algorithm to compute an upper bound on the margin of resilience when the routing policy only has access to the local state of the network. We identify sufficient conditions on network parameters under which the upper bound is tight under an appropriate routing policy.  
\end{abstract}
%
%\begin{keywords} 
%cascading failure, robustness, adversarial attacks, routing, network flows, infrastructure networks
%\end{keywords}

\section{Introduction}
\label{sec:introduction}
Resilience is becoming a key consideration in the design and operation of many critical infrastructure systems such as transportation, power, water, and data networks. Due to their increasing scale and interconnectedness, these systems tend to exhibit complex behaviors that pose several new challenges in their design and operation. 
In particular,  while exhibiting good trade-off between performance and robustness in the business-as-usual regime, such systems could be fragile to small local disruptions that may give rise to  cascading failures with potentially systemic effects. The problem is further exacerbated by the fact that local actions aimed at mitigating disruptions can increase vulnerability of the other parts of the system.

Models for cascading phenomena in infrastructure networks have been proposed in the statistical physics literature and studied mainly through numerical simulations, e.g., see \cite{Cohen.Erez.ea:00,Watts:01,Motter:2002fk,Crucitti.Latora.ea:04,Barrat.Barthelemy.ea:08}. Simpler models, based on percolation and other interacting particle systems describing the activation status of nodes and links as dependent on the activation status of their neighbors in the network, have lend themselves to more analytical studies, \cite{Liggettbook,Percolation, Adler:91}. While largely used to model the spread of epidemics and rumors in social and economic networks, cascading failures in financial networks and in wireless networks \cite{VegaRedondoCSN,DraiefMassoulie,Lelarge,Kong.Yeh:10}, the applicability of the latter models to the design and control of actual physical networks is severely limited because of their simplistic description of the causal relationship between failures of successive nodes and links. In particular, an inherent characteristic of such percolation- and interacting particles-based models is that the successive nodes and links to fail are constrained to be adjacent to each other, which is typically not the case in infrastructure networks. (see, e.g., \cite{Bernstein.ea:11}) Recently, more physically motivated dynamical models for cascading failures overcoming such limitations have been proposed and analyzed in the context of power networks \cite{Dobson.Carreras.ea:05,Lai.Low:Allerton13}, even allowing for control in between successive failure events~\cite{Bienstock:11}. 

This paper is concerned with  dynamical model for cascading failures in  single-commodity flow networks, and with the characterization of maximally resilient routing policies.
When considering dynamical models for cascading failures in physical infrastructure networks, there are several possibilities for time scale separation between link inactivation dynamics under overload, flow dynamics and reaction time of routing (control) policies that can simplify the analysis. For example, in the case of power networks, flow dynamics is typically much faster than the link inactivation dynamics and reaction time of control policies. The rate of information propagation among geographically distributed routing policies relative to the dynamics can add further complexity. In this paper, we focus our analysis on the limiting case when the rate of information propagation is slow (i.e., routing policies are distributed), and the link inactivation and flow dynamics under routing policies evolve at the same and much faster time scale. Our ability to analyze the dynamical model relies on identifying conditions under which the network state evolves monotonically. Irreversibility in link inactivation in our model 
naturally implies monotonicity in the link activation status. However, monotonicity in the link flows requires additional restrictions on the routing policy. We study these restrictions under \emph{flow monotonicity} and \emph{link monotonicity} which refer to the sensitivity of the action of a distributed routing policy with respect to changes in inflow (due to changes in the upstream part of the network) and activation status of outgoing links, respectively. %\kscomment{These notions are to be contrasted with monotonicity arguments in percolation theory for analysis of cascading phonemena in social networks.} \kscommentphantom{(additional details come here)}  

%\kscomment{The dynamics of cascading failures in physical infrastructure networks is characterized by a complex interplay of information exchange among distributed controllers, 
%speed of control action, and link buffer dynamics that governs the speed of propagation of congestion effects and the latency between the overload of a link and its failure. In this paper, we focus on the specific case when the speed of information propagation is much slower than the speed of control action and link buffer dynamics.}
%
%\kscommentphantom{Control action consists of two parts: (i) response to increase in inflow; and (ii) response to disturbances. In this paper we focus on routing policies that are oblivious to the presence of disturbance until the instance of link failures. Moreover, we assume that the control action with respect to the 
%
%in (ii) is physically constrained to be instantaneous. Therefore speed of control action corresponds to action with respect to the inflow. By fast (resp., slow) buffer dynamics, we also include small (resp., large) buffer capacity.}

The contributions of the paper are as follows. First, we propose a dynamical model for cascading failures in network flows and formally state the problem of designing maximally resilient routing policies. Second, we propose a backward propagation algorithm for computing an upper bound on the margin of resilience and to motivate the design of a maximally resilient routing policy. 
%\kscomment{This algorithm is inspired by dynamic programming, where the cost-to-go function is replaced by an appropriate resilience-to-go function.} \ksmargin{we do not talk about dynamic programming in the main body of the paper} 
Third, we introduce the properties of flow and link monotonicity for distributed routing policies, and show that these are sufficient conditions for the upper bound to be tight. Finally, we identify sufficient conditions on network parameters under which these monotonicity conditions are satisfied. 

The paper is organized as follows. In Section~\ref{sec:cascade-model}, we present our dynamical model for cascading failures in network flow under routing policies. Section~\ref{sec:main-results} contains main results of this paper in terms of upper bound computation, monotonicity conditions on the routing policies and sufficient conditions on network parameters to guarantee these monotonicity properties. The proofs of the main results are collected in Section~\ref{sec:proofs}. Section~\ref{sec:conclusions} provides concluding remarks. The appendix contains a technical lemma which is used heavily in the proofs in Section~\ref{sec:proofs}. 

%The rest of the paper is organized as follows. In Section~\ref{sec:cascade-model}, we describe our model for cascading failures in network flows and illustrate the differences with respect to the percolation models. Section~\ref{sec:problem-formulation} formulates the problem of resilient control design under the proposed framework. Section~\ref{sec:main-results} describes the backward propagation algorithm and uses it to establish an upper bound on the margin of resilience, and also proposes a specific class of distributed oblivious routing policies, and describes multiple scenarios under which it is provably maximally resilient. The proofs of the technical results in the paper are collected in Section~\ref{sec:proofs}. Finally, we conclude in Section~\ref{sec:conclusions} with remarks on future work. 

Before proceeding, we define some preliminary notations to be used throughout the paper. Let $\real$ be the set of real numbers, $\realnoneg:=\{x\in\R:\,x\ge0\}$ be the set of nonnegative real numbers, and $\N$ be the set of natural numbers. When $\mc A$ is a finite set, $|\mc A|$ will denote the cardinality of $\mc A$, $\R^{\mc A}$ (respectively, $\R_+^{\mc A}$) will stay for the space of real-valued (nonnegative-real-valued) vectors whose components are indexed  by elements of $\mc A$.
For $x \in \real^{\mc A}$ and $y \in \real_+^{\mc B}$, $x'$ stands for the transpose of $x$, and
$x \leq y$ means that $x_i \leq y_i$ for all $i \in \mc A \cap \mc B$. 
When $\mc A=\mc B$, $x'y$ stands for the dot product of $x$ and $y$. 
%, and $\R^{\mc A\times\mc B}$ for the space of matrices whose real entries indexed  by pairs of elements in $\mc A\times\mc B$. 
%The transpose of a matrix $M \in \real^{\mc A \times\mc B}$, will be denoted by $M^T \in\R^{\mc B\times\mc A}$, while 
The all-one and all-zero vectors will be denoted by $\onebf$ and $\zerobf$, respectively, their size being clear from the context. 
%\ksmargin{convert all indicator functions in the paper to $\mathbb{1}$} 
%\ksmargin{do we use indicator functions anymore ?}
%$\mathbb{1}_{\mc A}$ stands for the indicator function over set $\mc A$.
%Let $\cl(\mc X)$ be the closure of a set $\mc X\subseteq\R^{\mc A}$. 
A directed multigraph is the pair $(\mc V,\mc E)$ of a finite set $\mc V$ of nodes, and of a multiset $\mc E$ of links consisting of ordered pairs of nodes (i.e., we allow for parallel links between a pair of nodes). 
If $e=(v,w)\in\mc E$ is a link, where $v,w \in \mc V$, we shall write $\sigma_e=v$ and $\tau_e=w$  for its tail and head node, respectively. The sets of outgoing  and incoming links of a node $v\in\mc V$ will be denoted by $\mc E^+_v:=\{e\in\mc E:\,\sigma_e=v\}$ and $\mc E^-_v:=\{e\in\mc E:\,\tau_e=v\}$, respectively.
%\ksmargin{not sure if we are actually using these short hand notations}
%Moreover, we shall use the shorthand notation $\mc R_v:=\R_+^{\mc E^+_v}$ for the set of nonnegative-real-valued vectors whose entries are indexed  by elements of $\mc E^+_v$; 
%%for a  given $\mu \ge 0$, 
%%$\mc S_v(\mu):=\{x \in \mc R_v:\,\sum_{e \in \mc E_v^+} x_e=\mu\}$; 
%and $\mc R:=\R_+^{\mc E}$ for the set of nonnegative-real-valued vectors whose entries are indexed  by the links in $\mc E$. 
For $x \in \real$, we shall use the notation $[x]^+$ to mean $\max\{0,x\}$.

%\section{Problem Setup}
%\label{sec:setup}

\section{Dynamical Model for Network Flows and Problem Formulation}
\label{sec:cascade-model}
In this section, we propose a dynamical model for cascading failure in network flows under distributed routing policies. 
We model flow networks as finite weighted directed multi-graphs $\mc N=(\mc V,\mc E,C)$, where $\mc V$ and $\mc E$ stand for the sets of nodes and links, respectively, and $C\in\R^{\mc E}$ is the vector of link capacities, all assumed to be strictly positive. We refer to nodes with no incoming links as origin nodes and to those with no outgoing links as destination nodes. The set of destination nodes is denoted by $\mc D$. Nodes which are neither origin nor destination are referred to as intermediate nodes and are assumed to lie on a path from some origin to some destination. 

Let an external inflow $\lambda_o\ge0$ be associated to every origin node $o \in \mc V$, and, by convention, put $\lambda_v=0$ for every other node $v$. 
Then, the max-flow min-cut theorem, e.g., see \cite{Ahuja.Magnanti.ea:93},  implies that a necessary and sufficient condition for the existence of a feasible equilibrium flow is that the capacity of every cut in the network is larger than the aggregate inflow associated to the non-destination side of the cut. Here, a feasible equilibrium flow refers to a vector $f\in\R_+^{\mc E}$ satisfying capacity constraints $f_e<C_e$ on every link $e\in\mc E$, and mass conservation at every non-destination node, i.e, $\lambda_v+\sum_{e\in\mc E^+_v}f_e=\sum_{e\in\mc E^-_v}f_e$ for all $v\in\mc V\setminus\mc D$. On the other hand, a cut refers to a subset of non-destination nodes $\mc U\subseteq\mc V\setminus\mc D$ , with $C_{\mc U}:=\sum_{e\in\mc E:\sigma_e\in\mc U,\tau_e\in\mc V\setminus\mc U}C_e$ standing for its capacity and $\lambda_{\mc U}:=\sum_{v\in\mc U}\lambda_v$ for the associated aggregate external inflow. Then, the necessary and sufficient condition for the existence of a feasible equilibrium flow is  
\be\label{maxflowmincut}
\max_{\mc U}\left\{\lambda_{\mc U}-C_{\mc U}\right\}<0,
\ee
with the index $\mc U$ running over all possible cuts.

We now describe network flow dynamics, evolving in discrete time.  
Let $\mc N=(\mc V,\mc E,C)$ be a network as above, with inflows $\lambda_o$ at the origin nodes satisfying condition (\ref{maxflowmincut}). 
At every time $t=0,1,\ldots$, the state of the system is described by a tuple $(\mc V(t),\mc E(t),f(t),C(t))$ where:
$\mc V(t)\subseteq\mc V\setminus\mc D$ and
$\mc E(t)\subseteq\mc E$ are the subsets of active non-destination nodes, and links, respectively;
$f(t)\in\R_+^{\mc E}$ is the vector of link flows; 
and $C(t)\in\R^{\mc E}$, with $0\le C_e(t)\le C_e$, is the vector of residual link capacities. 
The initial condition $(\mc V(0),\mc E(0),C(0),f(0))$ is such that $\mc V(0)=\mc V\setminus\mc D$, $\mc E(0)=\mc E$, i.e., all non-destination nodes and all links start active, $C(0)=C$, and $f(0)$ is a feasible equilibrium flow for $\mc N$.

%\margin{GC: I guess if we write this way we should have the whole dynamics running from t=1...}
Given its current state $(\mc V(t),\mc E(t),f(t),C(t))$ at time $t=0,1, 2, \ldots$, the network evolves as follows. 
All currently active links which become overloaded, i.e., whose current flow exceeds the current residual capacity, along with all those whose head node is currently inactive,  become irreversibly inactive, i.e., 
\be\label{eq:E-dynamics} \mc E(t+1)=\mc E(t)\setminus\{e\in\mc E(t):\,f_e(t)\ge C_e(t)\}\setminus\{e\in\mc E(t):\,\tau_e(t)\notin\mc V(t)\}\,.\ee
All currently active nodes $v$ that have no active outgoing link become irreversibly inactive, i.e., 
\be\label{eq:V-dynamics} \mc V(t+1)=\mc V(t)\setminus\{v\in\mc V(t):\,\mc E^+_v(t)=\emptyset\}\,.\ee
At every currently active node $v\in\mc V(t)$, a routing policy determines how to split the current inflow $\lambda_v(t):=\lambda_v+\sum_{e\in\mc E^-_v(t)}f_e(t)$ among the set $\mc E^+_v(t)$ of its currently active outgoing links, so that 
\be\label{eq:routingupdate} f_e(t+1)=G_e\left(\mc E_v^+(t),\lambda_v(t)\right)\,,\qquad e\in\mc E^+_v(t)\,.\ee 
Finally, the residual capacity vector is reduced by a disturbance $\delta(t)\in\R_+^{\mc E}$ so that 
\be\label{eq:C-evolution}C_e(t+1)=C_e(t)-\delta_e(t+1)\,,\qquad e\in\mc E(t).
\ee
The sequence $(\delta(1), \delta(2),\ldots)\subseteq\R_+^{\mc E}$ of incremental flow capacity reductions is meant to represent an external, possibly adversarial and network state dependent, process that, without any loss of generality, will be assumed to satisfy 
\be\label{eq:Dcum-def}\Delta(t):=\sum_{1\le s \le t}\delta(s)\le C\,,\qquad \forall t\ge1\,.\ee 

Observe that, in writing (\ref{eq:routingupdate}), we have assumed that the routing at node $v$ is determined only by the local observation of the current inflow $\lambda_v(t)$ and the currently active set of outgoing links $\mc E^+_v(t)$. In particular, the routing policies have no information about the residual link capacities, or equivalently about the disturbance process. 
The formal definition of distributed oblivious routing policies is as follows. 
\begin{definition}\label{def:oblivious}
Given a network $\mc N=(\mc V,\mc E,C)$, a \emph{distributed oblivious routing policy} $\mc G$ is a family of functions 
$$G^v(\mc J, \,\cdot\,):\R_+\to\R_+^{\mc J}\,,\qquad v\in\mc V\setminus\mc D\,,\quad \emptyset\ne\mc J\subseteq\mc E^+_v\,,$$
such that, for every $\mu \ge0$, 
$\sum_{e \in\mc J}G^v_e(\mc J, \mu)=\mu$, 
and, for all $\mc K\subseteq\mc J\subseteq\mc E^+_v$, 
\be
\label{eq:link-monotonicity}
G^v(\mc J,\mu)\le G^v(\mc K,\mu).
%G^v_e(\mc J,\mu)\le G^v_e(\mc K,\mu)\,,\qquad \forall e \in\mc K\,.
\ee
\end{definition}
In reading \eqref{eq:link-monotonicity}, recall our notation established at the end of Section~\ref{sec:introduction} that, for $x \in \real^{\mc A}$ and $y \in \real_+^{\mc B}$, $x \leq y$ implies $x_i \leq y_i$ for all $i \in \mc A \cap \mc B$.
%\ksmargin{do we need to emphasize $G_e=0$ if $e \notin \mc J$ ?}
Definition~\ref{def:oblivious} implicitly implies that $G_e^v(\mc J,\mu)=0$ for all $e \in \mc E_v^+ \setminus \mc J$.
{Moreover, we will assume throughout that the initial equilibrium flow $f(0)$ is consistent with the given distributed oblivious routing policy, i.e., $G^v_e(\mc E_v^+, \lambda_v(0))=f_e(0)$ for all $e \in \mc E_v^+$, $v \in \mc V \setminus \mc D$. In other words, the initial equilibrium flow is specified by the routing policy and, as long as there is no perturbation, i.e., $\delta(t)=0$, the network state does not change.} 
The term \emph{oblivious} in distributed routing policies is meant to emphasize that routing policies 
have no information about the disturbance process. 
Hereafter, unless explicitly stated otherwise, we shall refer to a routing policy satisfying Definition~\ref{def:oblivious} simply as a distributed routing policy.
Equation \eqref{eq:link-monotonicity} implies that, at every node, for a fixed inflow, shrinking of the set of active links results in increase in flow assigned to each of the remaining active outgoing links. We shall refer to \eqref{eq:link-monotonicity} as the \emph{link monotonicity} property. 
While \eqref{eq:link-monotonicity} represents a natural condition for distributed routing policies, the maximally resilient routing policies designed in this paper have been found to satisfy it. Alternately, one could regard the results in this paper to be optimal within this class of distributed routing policies. 
We provide additional comments on this aspect in Remark~\ref{rem:link-monotonicity-justification}.

\begin{remark}
\label{rem:link-monotonicity}
\eqref{eq:link-monotonicity} is satisfied by any routing policy at a node $v$ if $|\mc E_v^+| \leq 2$. 
\end{remark}

A simple example of a distributed routing policy is the \emph{proportional} policy: 
for every $\emptyset \neq \mc J \subseteq \mc E_v^+$, $v \in \mc V \setminus \mc D$, $\mu \geq 0$: 
\begin{equation}
\label{eq:prop-routing}
G_{e}^v(\mc J,\mu)= 
\left\{\ba{lcl} 
\left(\flowmax_e / \sum_{j \in \mc J} \flowmax_j \right) \, \mu & \se & e \in \mc J, \\
0 & \se & e \notin \mc J\,.
\ea \right.
\end{equation}

The model in \eqref{eq:E-dynamics}-\eqref{eq:C-evolution} has several salient features worth emphasizing. First, note that the transition from active to inactive status of a link is irreversible. Second, note that a link could become inactive either because it is overloaded or because its downstream node becomes inactive. The mismatch between flow and residual capacity of a link, which gives rise to overload condition, depends on the disturbance process and the action of a distributed routing policy. Therefore, the links to fail successively are not necessarily adjacent to each other. 
%\kscomment{This feature distinguishes our model from other cascading failure models such as bootstrap percolation, e.g., see \cite{Adler:91}, threshold models, e.g., see \cite{Granovetter:78} and other commonly used models for diffusion and cascading phenomena on networks, e.g., see \cite{Watts:01,Crucitti.Latora.ea:04}. Hence, the analysis of this model is not amenable to percolation theory.}  
Finally, note that in our model, routing policy updates its action at the same time scale as flow and link inactivation dynamics. An implication of this is that the flow vector $f$ may not be an equilibrium flow at all time instants because of violation of flow conservation at some nodes. This is in contrast to the setting of power networks, where the time scale for flow dynamics is much faster than the link failure dynamics and control action. It is possible to extend \eqref{eq:E-dynamics}-\eqref{eq:C-evolution} to model scenarios representing a combination of centralized and non-oblivious routing, link recovery after finite time and long range coupling between failure of links. However, the analysis presented in this paper is restricted to the model in \eqref{eq:E-dynamics}-\eqref{eq:C-evolution}.

The following example illustrates cascading failure under the dynamics in \eqref{eq:E-dynamics}-\eqref{eq:C-evolution}. 
\begin{example}
\label{ex:cascading-failure}
Consider the graph topology depicted in Figure~\ref{fig:depth4}, where the flow capacities are given by $\flowmax_i=4$ for $i=1,2$, $\flowmax_i=3$ for $i=3,4,6,7,10$, $\flowmax_i=1.5$ for $i=5,9$ and $\flowmax_8=0.75$. Let the arrival rate at the origin be $\lambda=4$. 
\begin{figure}[htb!]
\begin{center}
%\vspace{0.1in}
\includegraphics[width=8cm]{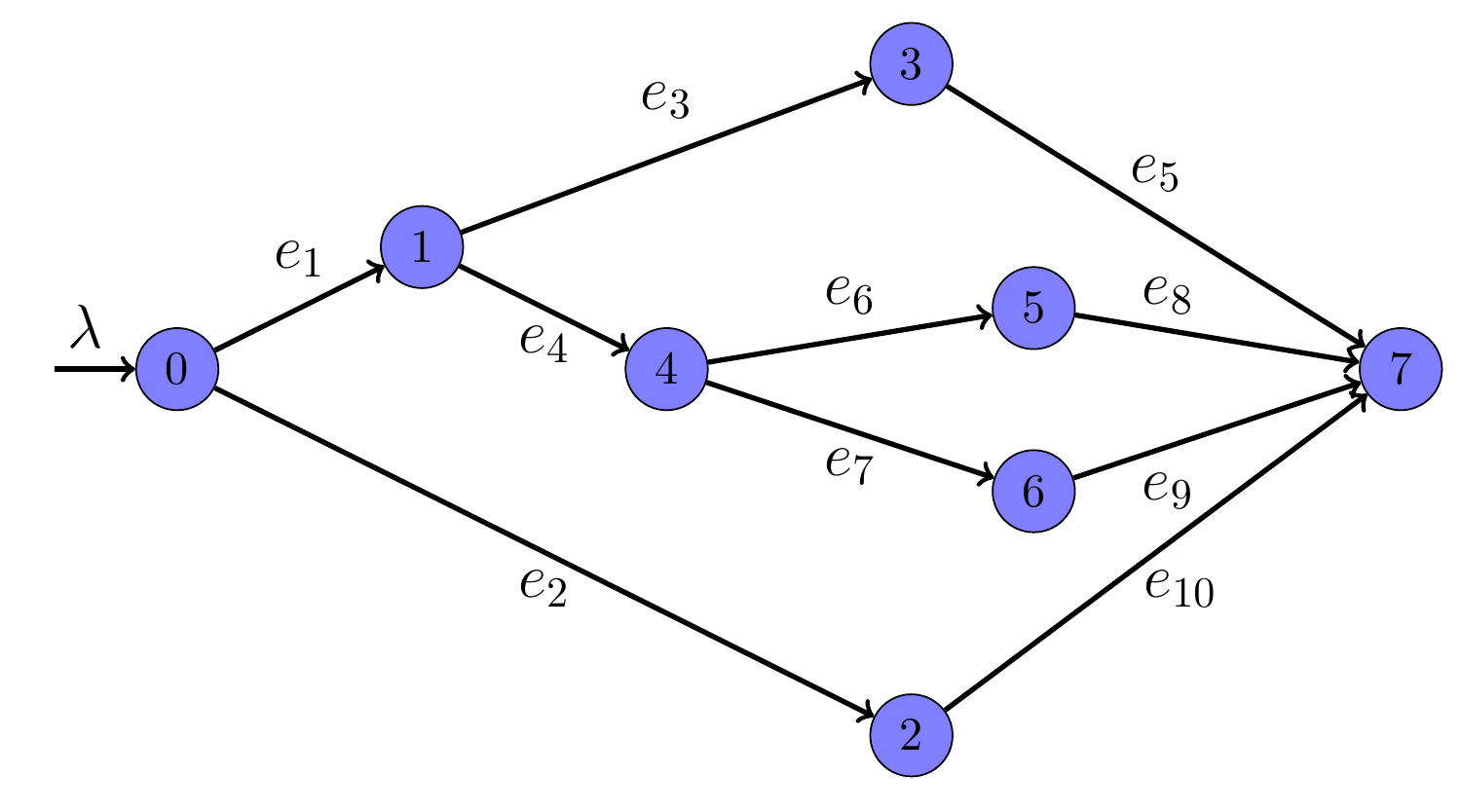}
\end{center}
\caption{A simple graph for the illustration of cascading failure under the proposed network dynamics.}
\label{fig:depth4}
\end{figure}
We consider proportional routing policies, specified in \eqref{eq:prop-routing}, at all the nodes, under which the initial flow on all links are given by $f_i(0)=2$ for $i=1,2,10$, $f_i(0)=1$ for $i=3,4,5$, and $f_i(0)=0.5$ for $i=6,7,8,9$. We now consider the network dynamics under a disturbance process for which $\delta_5(1)=0.55$, $\delta_e(t)=0$ for all $t \geq 2$ and $\delta_i(t) \equiv 0$ for all $i \in \until{10} \setminus \{5\}$. 
Since $\flowmax_5(1)=1.5-0.55=0.95<f_1(1)=1$, $e_5 \notin \mc E(2)$.
This is followed by $3 \notin \mc V(3)$ and $e_3 \notin \mc E(4)$. Therefore, $f_4(5)=2$, $f_6(6)=f_7(6)=f_8(7)=f_9(7)=1$. Since $\flowmax_8(7)=\flowmax_8(0)=0.75 < f_8(7)$, $e_8 \notin \mc E(9)$.
%
%For these values, we have that $\chi_{e_{5}}(0)=0$. This implies that $\xi_{e_5}(0)=0$, and hence $\psi_3(0)=0$. Therefore, $\xi_{e_3}(1)=0$, which implies that $f_{e_4}(2)=2$, $f_{e_6}(3)=f_{e_7}(3)=1$ and $f_{e_8}(4)=f_{e_9}(4)=1$. Since $\flowmax_8=0.75<1=f_{e_8}(4)$, this implies that $\chi_{e_8}(4)=0=\xi_{e_8}(4)=0$. 
%
By continuing along these lines, the order of links to become inactivated is $e_5,e_3,e_8,e_6,e_9,e_7,e_4,e_1,e_{10},e_2$. This clearly demonstrates that the links to fail successively under our proposed network dynamics are not necessarily adjacent to each other.
\end{example}

\begin{remark}
\label{rem:previous-work}
The model in \eqref{eq:E-dynamics}-\eqref{eq:C-evolution} is to be contrasted with the dynamical flow network formulation in our previous work~\cite{Como.Savla.ea:Part1TAC10,Como.Savla.ea:Part2TAC10} where every link has infinite buffer capacity, and hence there are no cascading effects under link overload. This feature is relaxed in our subsequent work~\cite{Como.Lovisari.ea:TCNS13}, where the links are modeled to have finite buffer capacity, and the control policy at every node implements routing as well as flow control under information about the densities and the disturbances on the links incoming and outgoing from that node. 
Such a framework allows for backward cascade effect, which was proven to increase the resilience of the network with respect to the framework in  \cite{Como.Savla.ea:Part1TAC10,Como.Savla.ea:Part2TAC10}. Such control policies were also shown to exhibit \emph{graceful collapse}, i.e., when the inflow to the network exceed its capacity, then all the critical links saturate simultaneously. In this paper, we constrain the actions of the control policies to only routing, and under no information about the disturbance. We emphasize that, although the routing policies have no explicit information about disturbance on the links, they have information about its effect on the activation status on the local links. On the other hand, due to cascade effects, the change in the activation status of a link may not be exclusively due to disturbance on that link. 
\end{remark}

\subsection{Problem Formulation}
%%The problem formulation and analysis for the rest of the paper is restricted to networks with a single origin destination pair. 
%%While extensions to multiple destinations is straightforward, extensions to multiple origin nodes is not trivial. 
%%%\subsection{Quantifying Margin of Resilience}
%%We identify the node set $\mc V$ with $\zerountil{n}$, with $0$ and $n$ corresponding to the unique origin and destination nodes, respectively. 
%For brevity in notation, we let the external inflow associated with the
% unique origin node as $\lambda$.
In this paper, the performance criterion of interest is the ability of a network to transfer flow from the origin nodes to the destination nodes, under a wide range of disturbance processes. We formalize this notion as follows. 
%
%\kscommentphantom{(also assume that without loss of generality the initial condition of the network is at equilibrium under routing policy ?)}
%The concept of network transferring can now be formalized as follows.
\begin{definition}
%[Transfer efficiency of the network]
\label{def:transfer}
%Let $\mc N$ be a network, $\lambda$ a vector of inflows at the origin nodes, $\mc G$ a distributed routing policy, and $(\delta(t))_{t\ge1}$ a disturbance process. 
%Then, the associated network flow dynamics in \eqref{eq:E-dynamics}-\eqref{eq:C-evolution} is said to be \emph{transferring} if there exists $t^*\ge0$ such that 
%\be\label{eq:transferring}\sum_{d\in\mc D}\sum_{e\in\mc E^-_{d}}f_e(t)=\sum_{v}\lambda_v\,, \ee
%for all $t\ge t^*$. 
Let $\mc N$ be a network, $\lambda$ a vector of inflows at the origin nodes, $\mc G$ a distributed routing policy, and $(\delta(t))_{t\ge1}$ a disturbance process. 
Then, the associated network flow dynamics in \eqref{eq:E-dynamics}-\eqref{eq:C-evolution} is said to be \emph{transferring} if  
\be\label{eq:transferring} \lim_{t \to + \infty} \sum_{d\in\mc D}\sum_{e\in\mc E^-_{d}}f_e(t)=\sum_{v}\lambda_v\,,\ee
where the summation in $v$ is over the origin nodes.
\end{definition}

%Definition~\ref{def:transfer} readily extends to any (i.e., not necessarily distributed oblivious) routing policy.  
%Note that we do not make distinction between centralized and distributed routing policies in Definition~\ref{def:transfer}. 
Observe that, since $f(0)$ is assumed to be a feasible equilibrium flow, one has that, at time $0$, the aggregate outflow from and inflow to the network match, i.e., $\sum_{d\in\mc D}\sum_{e\in\mc E^-_{d}}f_e(t)=\lambda$ for $t=0$. 
Definition~\ref{def:transfer} requires  that, for a network $\mc N$ and a distributed routing policy $\mc G$ to be transferring under a disturbance process $(\delta(t))_{t \geq 1}$,
aggregate inflow in and outflow from the network also match asymptotically. For disturbance processes that are active only over finite time, \eqref{eq:transferring} can be rephrased to require the inflow and the outflow to match at all times with the possible exception of a finite transient. We shall use this latter formulation in Section~\ref{sec:main-results}, where the setup allows to focus only on finite time disturbance processes without loss of generality.  

%In particular, if the disturbance is such that $\delta(t) \equiv \zerobf$ for all $t \geq \tau \geq 0$, then the notion of network transferring is equivalent to requiring the inflow and the outflow to continue to match over time with the possible exception of a finite transient. 

The magnitude of a disturbance process $\delta$ is defined as (see \eqref{eq:Dcum-def}): $$\Dnorm(\delta):=\sum_{e \in \mc E} \Dcum_e(\infty).$$ 

\begin{definition}
%[Margin of resilience of the network]
\label{def:margin}
Let $\mc N$ be a network, $\lambda$ a vector of inflows at the origin nodes, and $\mc G$ a distributed routing policy. 
The margin of resilience of the network, denoted as $\mor(\mc N,\lambda,\mc G)$, is defined as the infimum of the magnitude of disturbance processes under which the associated dynamics is not transferring, i.e., 
$$\mor(\mc N,\lambda,\mc G):=\inf_{\delta} \setdef{\Dnorm(\delta)}{\text{network flow dynamics in \eqref{eq:E-dynamics}-\eqref{eq:C-evolution} for }\mc N, \lambda, \mc G, \delta \text{ is not transferring}}.$$
\end{definition}

%\ksmargin{not sure if this is the right place to put these graph-theoretic notations}
%We now recall a few additional graph-theoretic notions for trees. 
%\begin{definition}[Binary Tree]
%A tree $\mc N$ is called binary if $|\mc E_v^+| \leq 2$ for all $v \in \mc V$.
%\end{definition}
%
%\begin{definition}[Depth of a Tree]
%The depth of tree $\mc N$ is the length of the longest path in $\mc N$.
%\end{definition}
%

%\subsection{Problem Statement}
We are now ready to formally state the problem. 
Our objective in this paper is to (i) compute the margin of resilience under distributed routing policies; and (ii) identify maximally resilient distributed routing policies. 
Formally, we consider the following optimization problem: 
\begin{equation}
\label{eq:resilience-optimization-formulation}
\mor^*(\mc N,\lambda)=\sup_{\mc G} \mc R(\mc N,\lambda,\mc G),
\end{equation}
where the supremum is over the class of distributed routing policies. 
%A secondary objective is to find an oblivious routing policy $\tilde{\mc G}$ such that $\mc R(\lambda,\mc N,\tilde{\mc G})$ matches $\mc R^*(\lambda,\mc N)$ as closely as possible. 
A distributed routing policy $\mc G$ is called \emph{maximally resilient} if $\mc R(\mc N, \lambda, \mc G) = \mc R^*(\mc N,\lambda)$. 

\section{Main Results}
\label{sec:main-results}
In this section, we present our main results addressing problem \eqref{eq:resilience-optimization-formulation}. 
From now on, we will be restricted to networks $\mc N=(\mc V,\mc E,C)$ with a single origin destination pair. We will identify the node set $\mc V$ with the integer set $\{0,1,\ldots,n\}$, with $0$ and $n$ associated with the unique origin and destination nodes, respectively. Moreover, let $\lambda > 0$ be the constant inflow at the unique origin node. 
While extensions to multiple destinations is straightforward, extensions to multiple origin nodes is not trivial. We start by giving simple bounds on the margin of resilience.
\subsection{Simple Bounds}
It is straightforward to obtain the following upper and lower bounds on the margin of resilience, valid for every routing policy $\mc G$
\begin{equation}
\label{eq:simple-bounds}
\min_{e \in \mc E} \left\{\flowmax_e - f_e(0) \right\} \leq \mc R(\mc N,\lambda,\mc G) \leq \min_{\U} C_{\U} - \lambda\,,
\end{equation}
where the minimization in the upper bound is over all the cuts in $\mc N$. The lower bound in \eqref{eq:simple-bounds} is due to the fact that at least one link needs to become inactive to ensure non transferring of the network, possibly under cascading failure, and $\min_{e \in \mc E} \left(\flowmax_e - f_e(0) \right)$, which is the minimum among all link residual capacities, corresponds to the disturbance process with minimum magnitude that can cause a link to become inactive. The upper bound in \eqref{eq:simple-bounds}, which is usually referred to as the network residual capacity, is obtained by noting that the network is non-transferring under a disturbance process that removes residual capacity at $t=1$ from the links that constitute a min cut of $\mc N$. As it may be expected, the gap between the upper and lower bounds can be arbitrarily large in general networks. As an illustration, in Example~\ref{ex:cascading-failure}, the minimum link residual capacity is $0.25$, corresponding to link $e_8$, and the network residual capacity is $2.75$, corresponding to the cut $\{3,5,6,2\}$. However the example also constructs a disturbance process of magnitude $0.55$ under which the network is not transferring (under proportional routing policy). 
%Such large gaps between the trivial bounds in \eqref{eq:simple-bounds} motivate the objectives of this paper.

We now describe a recursive procedure to compute a sharper upper bound. The quantity computed by this procedure can be intuitively related to the margin of resilience when there is a clear time-scale separation between the flow dynamics (fastest), speed of control action (intermediate) and the link inactivation dynamics (slowest), and the routing policy is \emph{centralized} but oblivious to disturbance. In our current setup, it will provide an upper bound. Let $\mc X(\mc J, \lambda)$ be the set of equilibrium flow vectors when the inflow at the origin node is $\lambda$ and the set of active links is $\mc J \subseteq \mc E$, i.e., the set of $f\in\R_+^{\mc J}$ satisfying $f_e<C_e$ on every link $e\in\mc J$, $\sum_{e \in \mc E_0+ \cap \mc J} f_e = \lambda$ and $\sum_{e\in\mc E^+_v \cap \mc J}f_e=\sum_{e\in\mc E^-_v \cap \mc J}f_e$ for all $v\in\mc V\setminus \mc \{n\}$.
%
%We first formally define the set of feasible flow vectors. 
%For $\mc J \subseteq \mc E$ and $\lambda \geq 0$, let 
%\begin{multline}
%\label{eq:X-def}
%\mc X(\mc J,\lambda):=\Big\{x \in \real_+^{\mc E} \, : \, x_e \in [0,\flowmax_e] \, \, \forall \, e \in \mc J; \, x_e = 0 \, \, \forall \, e \in \mc E \setminus \mc J; \\ \sum_{e \in \mc E_v^- \cap \mc J} x_e = \sum_{e \in \mc E_v^+ \cap \mc J} x_e \, \forall \, v \in \mc V \setminus \{n\}; \, \, \sum_{e \in \mc E_v^+ \cap \mc J} x_e = \lambda \Big\}
%\end{multline}
%be the set of feasible flows over $\mc N$, and hence the action set of a centralized routing policy, when the set of active links is $\mc J$, and the inflow at the origin is $\lambda \geq 0$. 
%
Let $S(\mc J, \lambda)$ correspond to the margin of resilience when the network starts with active link set $\mc J \subseteq \mc E$. The computation of $S(\mc E, \lambda)$, which corresponds to margin of resilience of interest, is based on values of $S(\mc J, \lambda)$ for all $\mc J \subseteq \mc E$ as follows. 
Starting with sets $\mc J \subseteq \mc E$ of size one, i.e., $|\mc J|=1$, and then increasing in size, perform the following recursion: 
if $\mc X(\mc J, \lambda)=\emptyset$, then $S(\mc J,\lambda)=0$, else
\begin{equation}
\label{eq:centralized-routing}
S(\mc J, \lambda):=\max_{x \in \mc X(\mc J,\lambda)}  \min_{e \in \mc J} \Big(C_e - x_e + S\left({\mc J \setminus \{e\}}, \lambda \right) \Big).
\end{equation}
In \eqref{eq:centralized-routing}, $\flowmax_e-x_e$ is the difference between the original capacity of link $e$ and the flow on its under control action $x$ by a (centralized) routing policy. Therefore, it represents the minimum disturbance on link $e$ that will make it inactive. Since $S\left({\mc J \setminus \{e\}}, \lambda \right)$ represents the margin of resilience once link $e$ is removed, the whole term inside the minimum in \eqref{eq:centralized-routing} is the magnitude of the smallest disturbance required to make the network with $\mc J$ links to become non-transferring staring with removal of link $e$. The minimization in \eqref{eq:centralized-routing} over all $e \in \mc J$ searches for the initial link $e$ whose inactivation will minimize the disturbance required to make the network non-transferring. The outer maximization over the feasible action set of the (centralized) routing policy maximizes the magnitude of the worst-case disturbance process that will make the network non-transferring. The next results states the the output of the iterations in \eqref{eq:centralized-routing} is indeed an upper bound on the margin of resilience under any distributed routing policy. 

%An illustration of \eqref{eq:centralized-routing} is provided in Section~\ref{sec:simple-settings}. 

%For $\mc J \subseteq \mc E$, $\rlb \in \real_+^{\mc E}$ and $\mu \geq 0$, let $\mc X(\mc J,\rlb,\mu):=\setdef{x \in \real_+^{\mc E}}{x_e \in [\rlb_e,\flowmax_e] \, \, \forall \, e \in \mc J; \, x_e = 0 \, \, \forall \, e \in \mc E \setminus \mc J; \, \, \sum_{e \in \mc J} x_e = \mu}$ be the set of feasible flows over $\mc N$, and hence the action set of a centralized routing policy, when the set of active links is $\mc J$, the inflow at the origin is $\mu \geq 0$, and which is lower bounded element-wise by $\rlb$ on active links. Starting with $\mc J \subseteq \mc E$ having $|\mc J|=1$, and then increasing in size, perform the following recursion: 
%if $\mc X(\mc J, \rlb, \mu)=\emptyset$, then $S(\mc J,\rlb,\mu)=0$, else
%\begin{equation}
%\label{eq:centralized-routing}
%S(\mc J, \rlb, \mu):=\max_{x \in \mc X(\mc J,\rlb,\mu)}  \min_{e \in \mc J} \Big\{[C_e - x_e]^+ + S\left({\mc J \setminus \{e\}}, x, \mu\right) \Big\}.
%\end{equation}
%Note that the inclusion of lower bound $\rlb$ in the feasible set $\mc X(\mc J, \rlb, \mu)$ is to ensure the link monotonicity condition similar to \eqref{eq:link-monotonicity}. 
\begin{proposition}
\label{prop:centralized-upper-bound}
Let $\mc N=(\mc V,\mc E, C)$ be a network with $\lambda  > 0$ a constant total outflow at the origin node. Then, for any distributed routing policy, there exists a disturbance process $(\delta(t))_{t\ge1}$ with $\Dnorm(\delta) \leq S(\mc E,\lambda)$ under which the associated network flow dynamics \eqref{eq:E-dynamics}-\eqref{eq:C-evolution} is not transferring.
\end{proposition}

\begin{remark}
\label{rem:centralized-routing}
\begin{enumerate}
\item It is easy to prove that $S(\mc E, \lambda)$ is a tighter upper bound than the one in \eqref{eq:simple-bounds}.
\item The computation of $S(\mc E, \lambda)$ involves 
recursion over all sets $\mc J$ in $2^{\mc E}$. However, for each $e \in \mc J \subseteq \mc E$ and $\lambda \geq 0$, the term inside the minimization in \eqref{eq:centralized-routing} is affine in $x \in \mc X(\mc J, \lambda)$. Hence,  computing $S(\mc J, \lambda)$ is a convex optimization problem. 
\end{enumerate}
\end{remark}

For Example~\ref{ex:cascading-failure}, by simulations, we find that $S(\mc E, \lambda)=1.14$, which is less than $2.75$, the value corresponding to the upper bound in \eqref{eq:simple-bounds}. However, it is still greater than $0.55$, the magnitude of disturbance in Example~\ref{ex:cascading-failure}, under which the network is not transferring. This conservatism is because the recursions in \eqref{eq:centralized-routing} implicitly assume centralized routing, and do not take into account the possibility of link inactivation due to the inactivation of the corresponding head node. In Section~\ref{subsec:BPA}, we propose an algorithm, the Backward Propagation Algorithm (BPA), that addresses these limitations to provide a tighter upper bound, and we identify conditions under which this upper bound is provably tight. 
The BPA is designed for network topologies satisfying the following acyclicity assumption. 

\begin{assumption}
\label{ass:acyclicity}
$(\mc V, \mc E)$ contains no cycles.
\end{assumption}

A consequence of Assumption~\ref{ass:acyclicity}, the oblivious property of routing policies and the finiteness of $\mc V$ and $\mc E$ is that, we can assume without loss of generality that, for every $e \in \mc E$, there exists at most one $t_e \geq 0$ such that $\delta_e(t_e) > 0$, and that $\delta(t)=\zerobf$ after some finite time. %\margin{GC: I don't see why}
Therefore, it is sufficient to restrict our attention to disturbance processes $\delta$ that are non-zero only for a finite time, and hence there exists a finite time after which $(\mc V(t), \mc E(t), f(t), \flowmax(t))$ comes to a steady state under any such disturbance process $\delta$. Let $\mc T$ denote that finite termination time. 
In this case, Definition~\ref{def:transfer} simplifies as: network flow dynamics is  transferring if $\lambda_n(\mc T)=\lambda$.

The formulation and analysis of the BPA implicitly relies on the following simple result showing an equivalence between a network being transferring and its origin node being active all the time.

\begin{proposition}
\label{prop:transferring-notion-equivalence}
Let $\mc N$ be a network satisfying Assumption~\ref{ass:acyclicity} with $\lambda$ a constant inflow at the origin node, $\mc G$ a routing policy, and $(\delta(t))_{t\ge1}$ a disturbance process. 
Then, the associated network flow dynamics \eqref{eq:E-dynamics}-\eqref{eq:C-evolution} is transferring if and only if $0 \in \mc V(\mc T)$. Moreover, $\lambda_n(\mc T) \in \{0, \lambda\}$.
\end{proposition}

\begin{remark}
%\begin{enumerate}
%\item Proposition~\ref{prop:transferring-notion-equivalence} holds true for any routing policy too. 
%\item 
The analyses of conventional models for cascading failure focus primarily on the connectivity of the residual graph $(\mc V(\mc T), \mc E(\mc T))$. For the setting of this paper, the proof of Proposition~\ref{prop:transferring-notion-equivalence} can be used to easily show that there exists a directed path from $0$ to $n$ in $(\mc V(\mc T), \mc E(\mc T))$ if and only if the associated network flow dynamics is transferring.
%\end{enumerate}
\end{remark}

\subsection{Simple settings}
\label{sec:simple-settings}
Before describing the BPA, we present results for the maximal margin of resilience and the maximally resilient routing policy in simple settings.
We use these calculations merely to motivate the key steps in the construction of BPA, and refer to  Theorems~\ref{thm:upper-bound} and \ref{thm:lower-bound} for their rigorous justification.

\begin{figure}[htb!]
\begin{center}
%\vspace{0.1in}
\begin{minipage}[c]{.3\textwidth}
\begin{center}
\includegraphics[width=3cm]{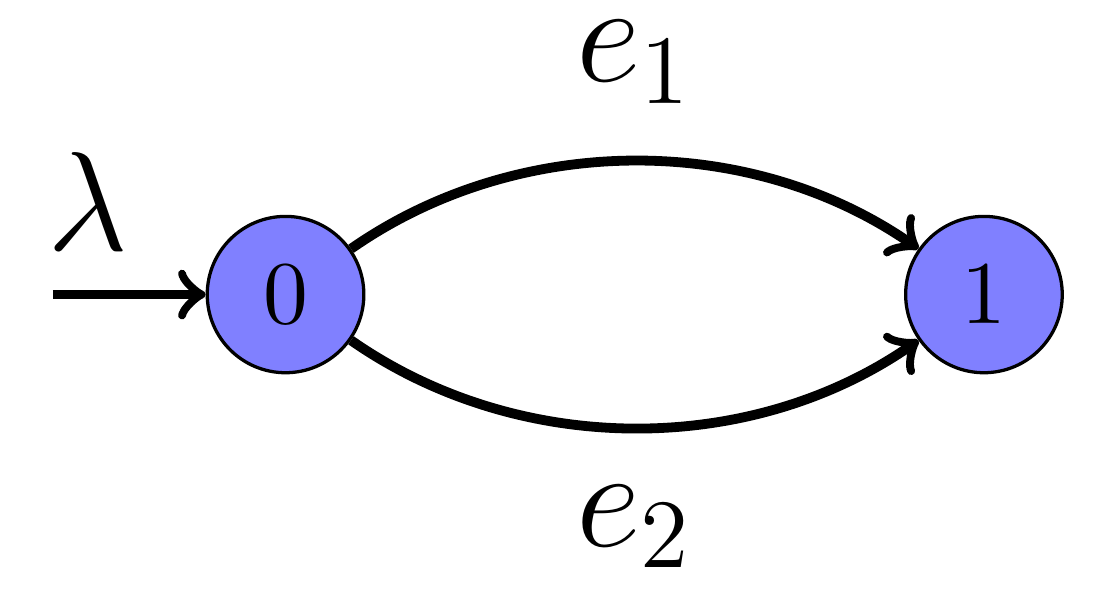}
\end{center}
\end{minipage}
\begin{minipage}[c]{.3\textwidth}
\begin{center}
\includegraphics[width=4.5cm]{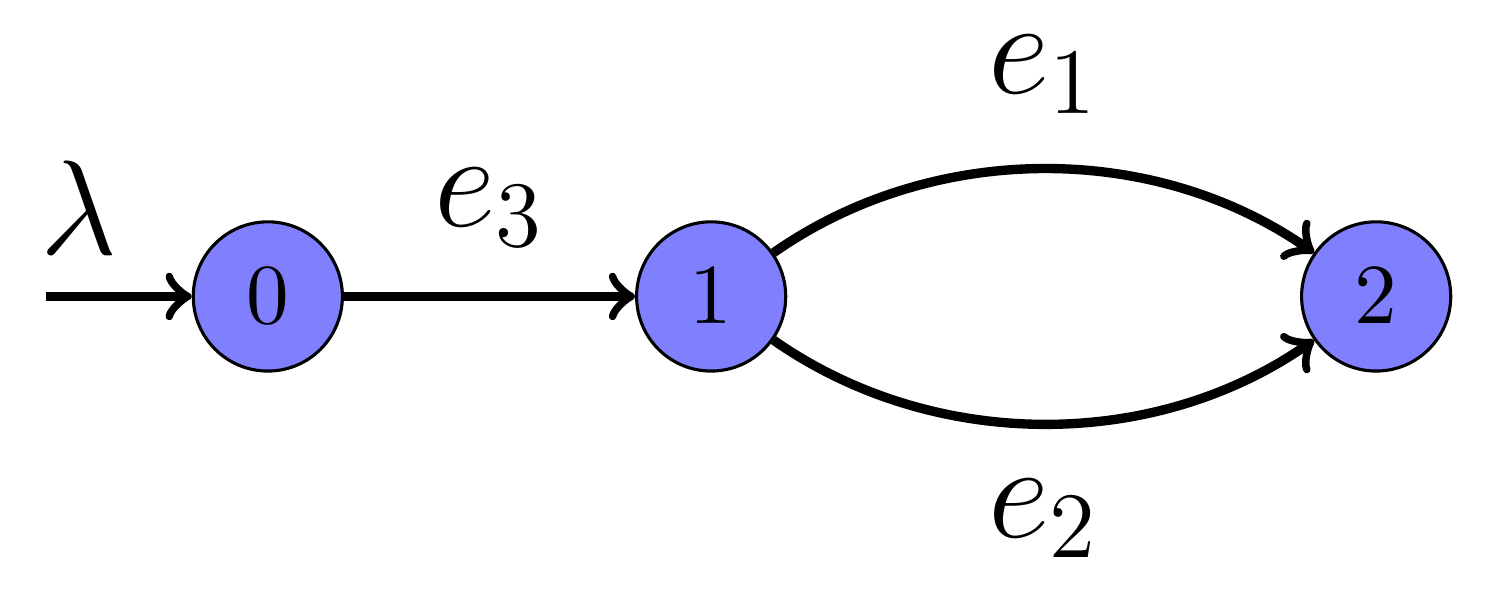} 
\end{center}
\end{minipage}
\begin{minipage}[c]{.3\textwidth}
\begin{center}
\includegraphics[width=3.5cm]{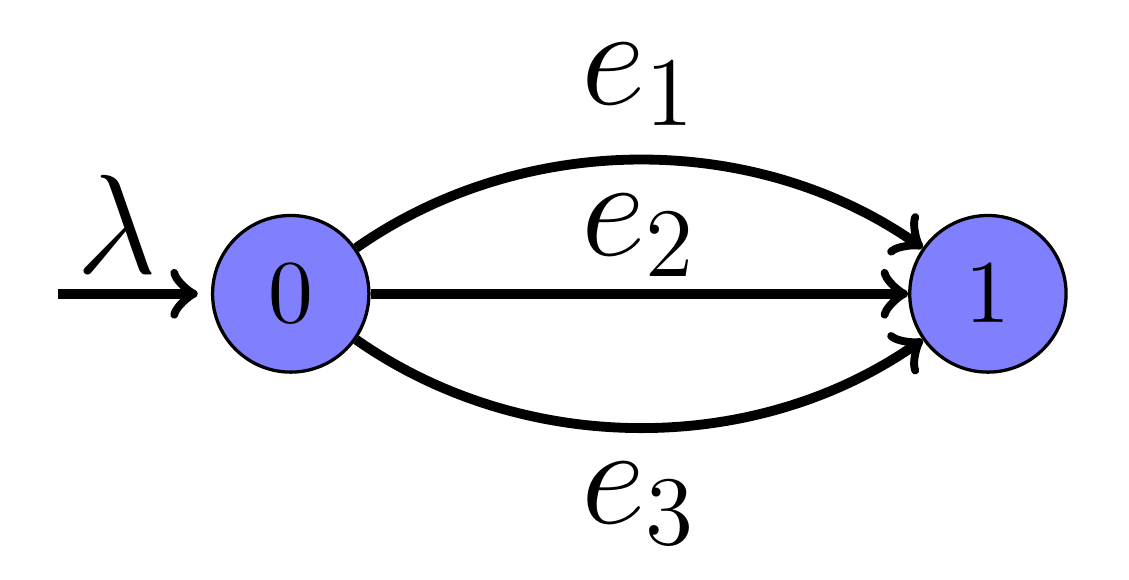} 
\end{center}
\end{minipage}

\begin{minipage}[c]{.3\textwidth}
\begin{center}
(a)
\end{center}
\end{minipage}
\begin{minipage}[c]{.3\textwidth}
\begin{center}
(b)
\end{center}
\end{minipage}
\begin{minipage}[c]{.3\textwidth}
\begin{center}
(c)
\end{center}
\end{minipage}
\end{center}

\caption{Illustrative simple network topologies.}
\label{fig:simple-setting}
\end{figure}

Let $\mc N_1$ denote the flow network illustrated in Figure~\ref{fig:simple-setting}(a), with $\mc E_1=\{e_1,e_2\}$ maximum link flow capacities $\flowmax_i$, $i=1,2$. Following Remark~\ref{rem:link-monotonicity}, the routing policy at node $0$ is completely specified by any vector $x \in \mc X(\mc E_1,\lambda)$. Considering all the possible outcomes of the disturbance process, the margin of resilience is given by:
%\be
%\label{eq:2nodes-mor}
%\mor(\mc N_1,\lambda, x) = \min\left\{\flowmax_1-x_1 + [\flowmax_2-\lambda]^+, \flowmax_2-x_2 + [\flowmax_1-\lambda]^+ \right\},
%\ee
\begin{align}
\label{eq:2nodes-mor}
\mor(\lambda,\mc N_1, x) & =\min\left\{[\flowmax_1-x_1]^+ + [\flowmax_2-\lambda]^+, [\flowmax_2-x_2]^+ + [\flowmax_1-\lambda]^+, [\flowmax_1-x_1]^+ + [\flowmax_2-x_2]^+ \right\} \\
\nonumber
& = \min\left\{[\flowmax_1-x_1]^+ + [\flowmax_2-\lambda]^+, [\flowmax_2-x_2]^+ + [\flowmax_1-\lambda]^+ \right\},
\end{align}
 where first term inside the $\min$ in the right hand side of \eqref{eq:2nodes-mor} corresponds to the inactivation of link $e_1$ at $t=2$ under $\delta(1)=\left[[\flowmax_1-x_1]^+ \, \, \, 0 \right]'$ followed by inactivation of $e_2$ at $t=3$ under $\delta(2)=\left[0 \, \, \, [\flowmax_2-\lambda]^+\right]'$, the second term corresponds to inactivation of $e_2$ at $t=2$ followed by inactivation of $e_1$ at $t=3$, and the third term corresponds to the simultaneous inactivation of links $e_1$ and $e_2$ at $t=2$ under $\delta(1)=\left[[\flowmax_1-x_1]^+ \, \, \, [\flowmax_2-x_2]^+ \right]'$.
  Therefore, the maximum possible margin of resilience, and the corresponding maximally resilient routing policy are, respectively, given by $\mor^*(\mc N_1,\lambda)=\max_{x \in \mc X(\mc E_1,\lambda)} \mor(\mc N_1, \lambda, x)$ and $x^*=\argmax_{x \in \mc X(\mc E_1,\lambda)} \mor(\mc N_1, \lambda, x)$, which, using simple algebra, can be computed as (see Figure~\ref{fig:Rv0} (a) for an illustration):
 \be
\label{eq:Rv-caseb}
\mor^*(\mc N_1, \lambda) = \left\{\ba{lcl} \flowmax_{1} + \flowmax_{2}-3\lambda/2 & \se & \lambda \in \left[0, \minCnode \right], \\ 
\minCnode/2+ \maxCnode-\lambda & \se & \lambda \in \left[ \minCnode, \maxCnode \right], \\
(\flowmax_{1}+\flowmax_{2})/2-\lambda/2 & \se & \lambda \in \left[ \maxCnode, \flowmax_{1} + \flowmax_{2}\right], \\
0 & \se & \lambda \geq \flowmax_{1} + \flowmax_{2} \,,
\ea \right.
\ee
where $\minCnode:=\min \left \{\flowmax_{1},\flowmax_{2} \right\}$ and $\maxCnode:=\max \left \{\flowmax_{1},\flowmax_{2} \right\}$, and (see Figure~\ref{fig:Rv0} (b) for an illustration)
\begin{equation}
\label{eq:Gstar1-caseb}
x^*_{1}(\lambda) = \left\{\ba{lcl} \lambda/2 & \se & \lambda \in \left[0, \minCnode \right] \\
\flowmax_{1}/2 + \left(\lambda - \left(\flowmax_{1}+\flowmax_{2}\right)/2 \right)  \mathbb{1}_{\flowmax_{1} > \flowmax_{2}} & \se & \lambda \in \left[\minCnode,
\maxCnode \right] \\
\lambda/2 + \left(\flowmax_{1}-\flowmax_{2}\right)/2 & \se & \lambda \in \left[\maxCnode, \flowmax_{1} + \flowmax_{2} \right].
\ea \r.
\end{equation}

\begin{figure}[htb!]
\begin{center}
\includegraphics[width=0.275\linewidth]{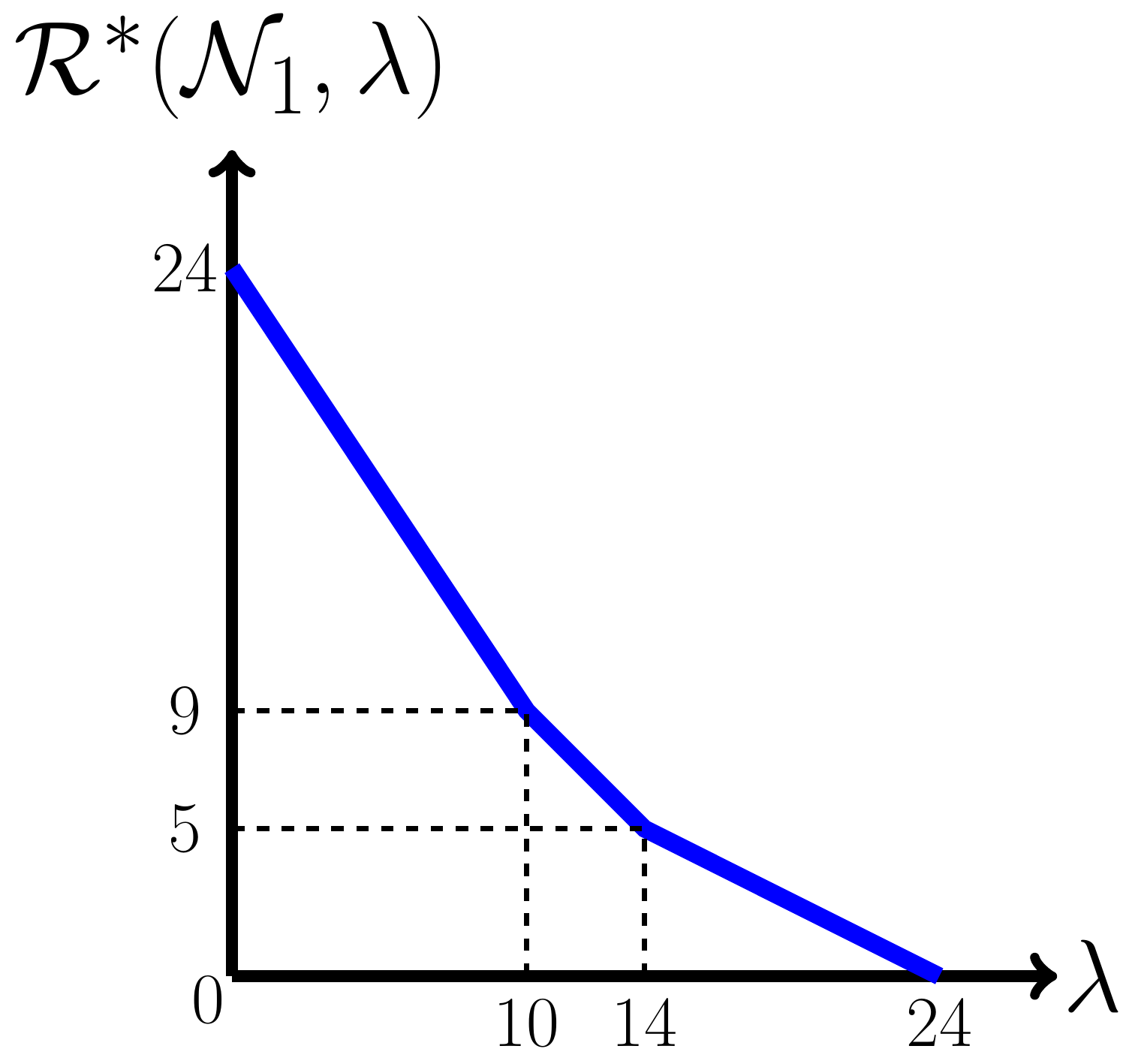} \hspace{0.2in}
\includegraphics[width=0.3\linewidth]{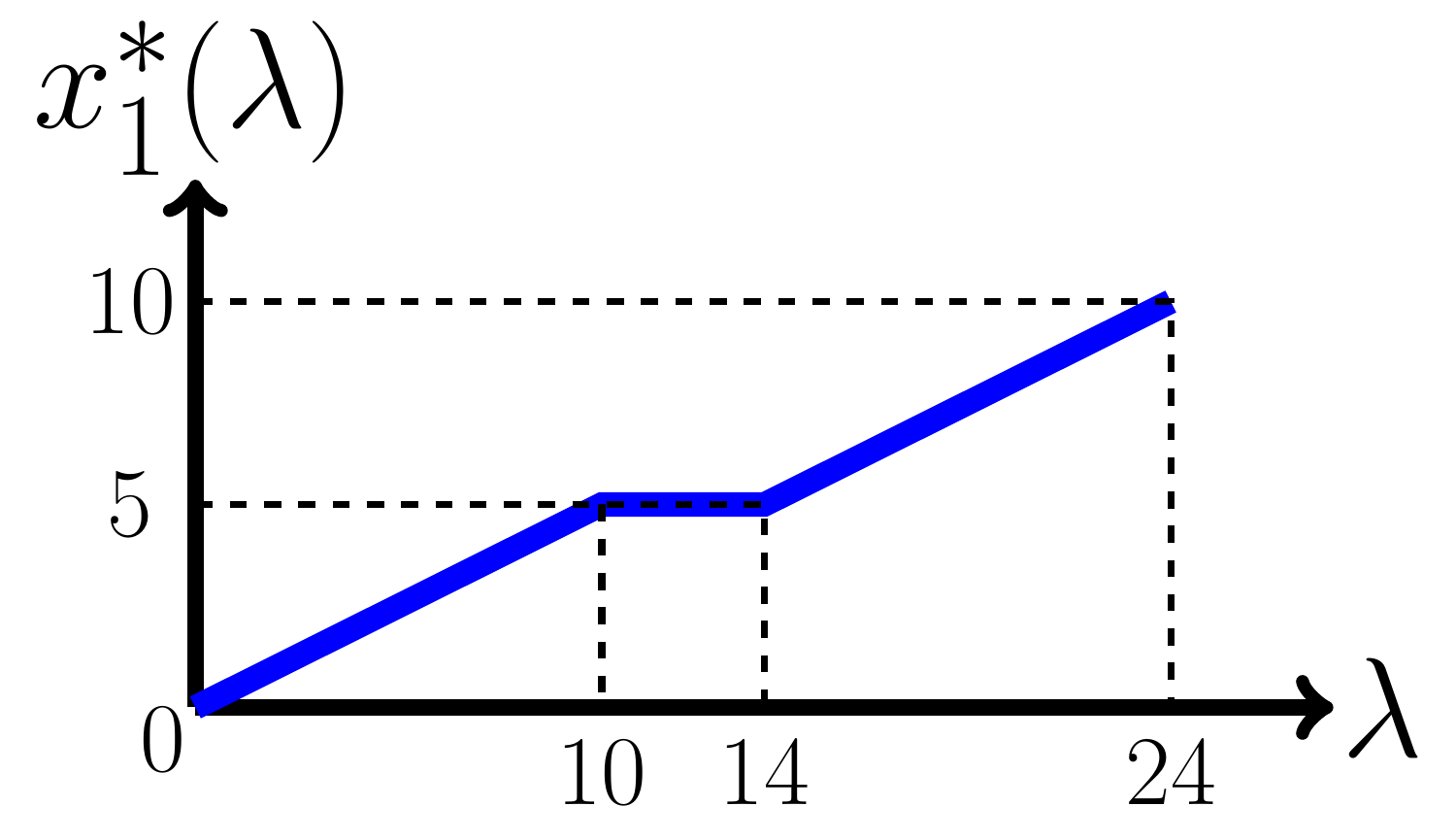} \hspace{0.2in}
\includegraphics[width=0.275\linewidth]{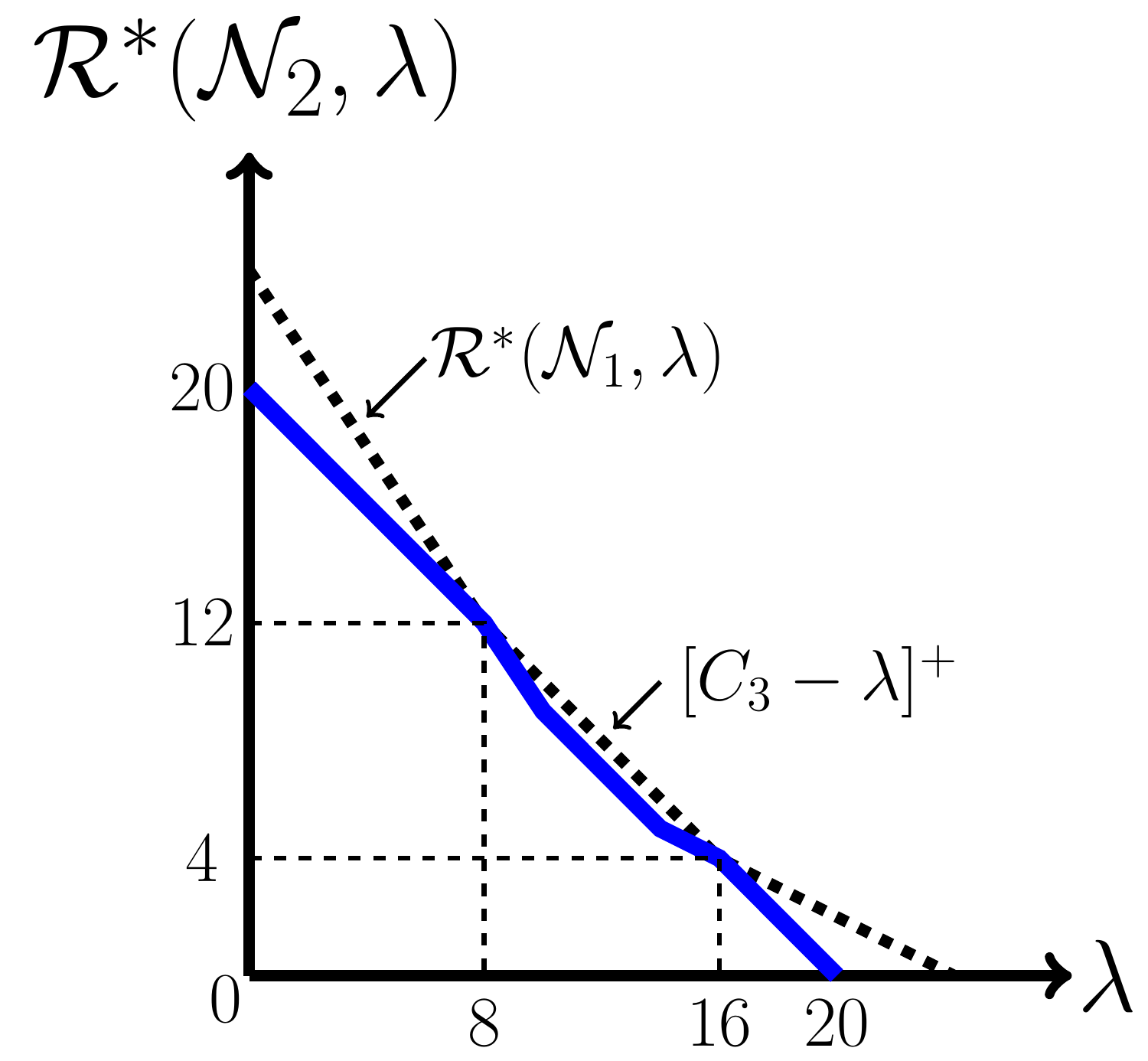} \\
(a) \qquad  \quad \qquad \qquad \qquad \qquad \qquad (b) \qquad  \qquad \qquad \qquad \qquad \quad \qquad (c)
\end{center}
\caption{Illustrations of (a) $\mc R^*(\mc N_1,\lambda)$ and (b) $x_1^*(\lambda)$ for $\flowmax_{1} = 10$ and $\flowmax_{2}=14$; and 
%For these parameters, $\minCnode =\flowmax_{1}=10$ and $\maxCnode=\flowmax_{2}=14$; 
(c) $\mc R^*(\mc N_2,\lambda)$ for $\flowmax_{3}=20$, $\flowmax_{1} = 10$ and $\flowmax_{2}=14$. 
%For these parameters, $\bar{\mu}_1=8$, $\bar{\mu}_2=16$, $\minCnode_{v_1}=\flowmax_{e_3}=10$ and $\maxCnode_{v_1}=\flowmax_{e_4}=14$.
}
\label{fig:Rv0}
\end{figure}
Figure~\ref{fig:Rv0} (b) illustrates that proportional routing policies (e.g., see \eqref{eq:prop-routing}) where the proportionality constants are independent of the arrival rate are in general not maximally resilient. 
For $\mc N_1$, comparing \eqref{eq:2nodes-mor} and \eqref{eq:centralized-routing}, we see that $\mor^*(\mc N_1,\lambda)=S(\mc E_1, \lambda)$. 

Let $\mc N_2$ denote the flow network illustrated in Figure~\ref{fig:simple-setting}(b), with maximum link flow capacities $\flowmax_i$, $i=1, 2, 3$. 
Following the dynamics in \eqref{eq:E-dynamics}-\eqref{eq:C-evolution}, link $e_3$ can become inactivate because of saturation of link $e_3$ or because of deactivation of node $\tau_{e_3}=1$. Accordingly, the maximum margin of resilience is the minimum of $[\flowmax_3-\lambda]^+$ and the maximum margin of resilience of the sub-network rooted at node $1$ when the inflow at node $1$. 
%given by
%$\mc R^*(\mc N_2,\lambda)=\min \left\{[\flowmax_3-\lambda]^+, R_{1}(\lambda) \right\}$, where $R_{1}(\lambda)$ is the maximum margin of resilience of the sub-network rooted at node $1$ when the inflow at node $1$ is $\lambda$. 
%
From our analysis of $\mc N_1$, the latter is equal to $\mc R^*(\mc N_1,\lambda)$, and hence the maximally resilient routing policy at node $1$ is the same as in \eqref{eq:Gstar1-caseb}. An illustration of $\mc R^*(\mc N_2,\lambda)$ is given in Figure~\ref{fig:Rv0} (c). In this case, $S(\mc E, \lambda)$ from \eqref{eq:centralized-routing} is in general not equal to $\mc R^*(\mc N_2, \lambda)$. This is because \eqref{eq:centralized-routing} does not take into account the fact link $e_3$ could become inactive due to inactivation of node $1$.

%The computation of maximal margin of resilience, and hence the maximally resilient routing policy requires consideration of all possible combinations of inactivation of the three links. We introduce a couple of notations to state the corresponding computation succinctly. 
%
%
Before proceeding with our next example, we define feasible flow vectors over active local links at node $v$. This will be the set of feasible control actions for the routing policy at node $v$. 
For $\mc J \subseteq \mc E_v^+$, $\rlb \in \real_+^{\mc E_v^+}$, $v \in \mc V \setminus \{n\}$ and $\mu \geq 0$, let
\be
\label{eq:X-v-def}
%\mc X_v(\mc J, \rlb, \mu):=\Big\{x \in \real_+^{\mc E_v^+} \, : \, x_e \in [r_e, \flowmax_e] \, \, \forall \, e \in \mc J; \, x_e = 0 \, \, \forall \, e \in \mc E_v^+ \setminus \mc J; \, \sum_{e \in \mc J} x_e = \mu \Big\}.
\mc X_v(\mc J, \rlb, \mu):=\Big\{x \in \real_+^{\mc J} \, : \, r \leq x \leq \flowmax; \, \onebf' x  = \mu \Big\}.
\ee
\eqref{eq:X-v-def} is a local version of the set $\mc X(\mc J, \lambda)$ used in \eqref{eq:centralized-routing} with two exceptions. First, \eqref{eq:X-v-def} is defined for a generic inflow $\mu$, since the inflow at $v$ is time-varying. Second, as will become clear in the construction of BPA, the presence of element-wise lower bound $\rlb$ allows one to impose the link monotonicity property defined in \eqref{eq:link-monotonicity}. 

Let $\mc N_3$ denote the flow network illustrated in Figure~\ref{fig:simple-setting}(c), with maximum link flow capacities $\flowmax_i$, $i=1, 2, 3$. Let $S(\mc J, \rlb, \lambda)$ be the margin of resilience when the set of active links at node $v$ is $\mc J \subseteq \mc E_v^+$, and when the action of routing policy is constrained to be element-wise greater than $\rlb$. For all $e \in \mc E_v^+$, $S(e,\rlb,\lambda)=0$ if $\lambda < r_e$ or $\lambda \geq \flowmax_e$, and $S(e,\rlb,\lambda)=C_e-\lambda$ otherwise. For $|\mc J| \geq 2$ in $\mc N_3$, one can write the following recursion:
\be
\label{eq:N3-recursion}
S(\mc J,\rlb,\lambda)=\max_{x \in \mc X_0(\mc J, \rlb, \mu)} \min_{e \in \mc J} \Big(\flowmax_e - x_e + S(\mc J \setminus \{e\}, x, \lambda) \Big).
\ee
Inside the minimization in \eqref{eq:N3-recursion}, the term $\flowmax_e-x_e$ is the difference  between the capacity of link $e$ when the flow on it is $x_e$, and hence represents the minimal disturbance required to make link $e$ inactive under routing action $x$. The term
$S(\mc J \setminus \{e\}, x, \lambda)$ is the magnitude of disturbance that is sufficient to make the network non-transferring after link $e$ has become inactive, under the constraint that the flows on links in $\mc J \setminus \{e\}$ can not be element-wise less than the flow $x$ on them when $e$ was active. The margin of resilience for $\mc N_3$ is then $S(\mc E, \zerobf, \lambda)$. The ability of \eqref{eq:N3-recursion} to incorporate link monotonicity constraints yields a sharper upper bound in comparison to \eqref{eq:centralized-routing}. 

The recursion in \eqref{eq:N3-recursion} can be used to derive margin of resilience for a network with arbitrary number of links between nodes $0$ and $1$ in Figure~\ref{fig:simple-setting} (c). However, in order to handle networks with arbitrary number of links between nodes $1$ and $2$ in Figure~\ref{fig:simple-setting} (b), we need to include the effect of inactivation of downstream nodes into \eqref{eq:N3-recursion}. This is the basis of the Backward Propagation Algorithm, which we describe next.

\subsection{The Backward Propagation Algorithm (BPA)}
\label{subsec:BPA}
We now describe the Backward Propagation Algorithm (BPA) to compute a tighter upper bound on the margin of resilience in comparison to Proposition~\ref{prop:centralized-upper-bound}. The same algorithm will also motivate the design of BPA routing which will be proven to be maximally resilient under certain sufficient conditions.

Assumption \ref{ass:acyclicity} implies that one can find a (not necessarily unique) topological ordering of the node set $\mc V=\zerountil{n}$ (see, e.g., \cite{Cormen.Leiserson:90}). We shall assume to have fixed one such ordering in such a way that $\mc E^-_{v}\subseteq\bigcup_{0\le u<v}\mc E^+_{u}$ for all $v=1,\ldots,n$.
%\begin{equation}
%\label{vertexordering}
%\mc E^-_{v}\subseteq\bigcup_{0\le u<v}\mc E^+_{u}\,,\qquad\forall  v=1,\ldots,n\,.\ee
%For $\mc J \subseteq \mc E$, let $\onebf_{\mc J}:=\setdef{x \in \{0,1\}^{\mc E}}{x_e = 1 \text{ if } e \in \mc J \text{ and } x_e=0 \text{ if } e \in \mc E \setminus \mc J}$.
%\ksmargin{do we need the definition of depth anymore?}
We recall that the \emph{depth} of a graph $(\mc V, \mc E)$ satisfying Assumption~\ref{ass:acyclicity} is the length of the longest directed path in $(\mc V, \mc E)$.

\begin{algorithm}[htb!]
\caption{The Backward Propagation Algorithm (BPA)}
\begin{algorithmic}[1]
\STATE $S(\mc E_n^+,\rlb,\mu):=+\infty$ for all $\rlb \in \real_+^{\mc E_n^+}$ and $\mu \geq 0$ \Comment{destination node}
\FOR {$v=n-1,n-2,\ldots ,0$} \Comment{construct a series of intermediate functions for every node starting with $n-1$, and going backward up to the origin}

\STATE for all $\rlb \in \real_+^{\mc E_v^+}$ and $\mu \geq 0$, 
$$S(\emptyset,\rlb, \mu) = 0$$
$$S(\mc J, \rlb, \mu) :=0 \text{ if } \mc X_v(\mc J, \rlb, \mu)=\emptyset, \quad \, \forall \, \, \emptyset \neq \mc J \subseteq \mc E_v^+,$$
%\STATE 
%\vspace{-0.15in}
%put ${S}(\emptyset,\rlb,\mu) = 0$ for all $\mu \geq 0$ and 
\begin{equation}
\label{eq:Se-def}
S_e(\mu)=S(e,\rlb,\mu):=\min\Big\{\flowmax_e-\mu, S(\mc E_{\tau_e}^+, \zerobf, \mu) \Big \} \quad \forall \, e \in \mc E_v^+.
\end{equation}

\STATE iteratively compute $S(\mc J, \rlb, \mu)$ for $\mc J \subseteq \mc E_v^+$ of increasing size, starting with sets of size 2: 
%\begin{equation}
%\label{eq:S-convention}
%{S}(\mc J, \rlb, \mu)=0  \quad \text{if } \mc X_v(\mc J, \rlb, \mu) = \emptyset\,,
%\end{equation}
%else  
\begin{equation}
\label{eq:S-tilde-def-new}
{S}(\mc J,\rlb,\mu)  :=\max_{x \in \mc X_v(\mc J, \rlb, \mu)} \, \, 
\min_{e \in \mc J} \Big( S_e(x_e) + S\big(\mc J \setminus \{e\}, x, \mu \big) \Big)
\end{equation}

\ENDFOR
\end{algorithmic}
\label{algo:BPA}
\end{algorithm}

Note that $\rlb$ appears only in the constraint set in the right hand side of \eqref{eq:S-tilde-def-new}. 
The fundamental difference between the recursions in \eqref{eq:S-tilde-def-new} and \eqref{eq:N3-recursion} is in the first term inside the minimization in \eqref{eq:S-tilde-def-new}. This term represents the minimum magnitude of disturbance required to make a link inactive. While it was sufficient to consider the disturbance on link $e$ for this purpose in $\mc N_3$, for general networks, \eqref{eq:Se-def} implies that the minimal disturbance could correspond to making the downstream node inactive. Therefore, the recursive computations at node $v$ depend on the outcome of the computations done for nodes downstream to $v$. 
The Backward Propagation Algorithm derives its name from the central feature of the algorithm, where an intermediate node collects $S(\mc J, \rlb, \mu)$ functions from its downstream nodes, performs updates with respect to local network parameters, and transmits it to upstream nodes. As such, the BPA can be executed in a distributed fashion. 

Complementary to the maximization in \eqref{eq:S-tilde-def-new} is the set of corresponding maximizers: 
\begin{equation}
\label{eq:g-def-1}
g\left(\mc J,\rlb, \mu \right) := \argmax_{x \in \mc X_v(\mc J, \rlb, \mu)} \, \, \min_{e \in \mc J} \Big( S_e(x_e) + S\left(\mc J \setminus \{e\}, x, \mu\right) \Big).
\end{equation}
A simple implication of \eqref{eq:g-def-1} which is used heavily in the paper is:
\be
\label{eq:g-def-implication}
z \in g\left(\mc J,\rlb, \mu \right) \implies z \geq r.
\ee

\subsection{Upper bound on the margin of resilience}
The quantity $S(\mc E_0^+,\zerobf,\lambda)$ computed by BPA is next shown to be an upper bound  on the margin of resilience under any distributed routing policy. For brevity in notation, we let $S^*(\mc N, \lambda):=S(\mc E_0^+,\zerobf,\lambda)$.
\begin{theorem}
%[Upper Bound]
\label{thm:upper-bound}
Let $\mc N$ be a network satisfying Assumption~\ref{ass:acyclicity} and with $\lambda$ a constant inflow at the origin node. Then, for any distributed routing policy $\mc G$, there exists a disturbance process $(\delta(t))_{t\ge1}$ with $\Dnorm(\delta) \leq S^*(\mc N, \lambda)$ under which the associated network flow dynamics \eqref{eq:E-dynamics}-\eqref{eq:C-evolution} is not transferring.
\end{theorem}

%\ksmargin{redefine $\mor^*$ to be maximization over all distributed policies}
\begin{remark}
\begin{enumerate}
\item Theorem~\ref{thm:upper-bound} implies that $\mor(\mc N,\lambda,\mc G) \leq S^*(\mc N, \lambda)$ for all distributed routing policies $\mc G$, and hence $\mor^*(\mc N,\lambda) \leq S^*(\mc N, \lambda)$.
\item While the statement of Theorem~\ref{thm:upper-bound} merely suggests the existence of a worst-case disturbance process, its proof in Section~\ref{subsec:ub-proof} explicitly constructs one such disturbance process. Therefore, in scenarios when the upper bound in Theorem~\ref{thm:upper-bound} is tight, this constructive procedure can also be used to identify the most vulnerable links of the network for adversarial setting.   

\item The computational complexity of BPA has tradeoffs in comparison to \eqref{eq:centralized-routing}. On one hand, while the recursion in \eqref{eq:centralized-routing} involves all elements in $2^{\mc E}$, BPA involves all elements only in $\cup_v 2^{\mc E_v^+}$, which is much smaller in comparison, especially when $|\mc V|$ is large.  On the other hand, \eqref{eq:centralized-routing} involves computation only for a fixed $\lambda$, whereas BPA involves computations, in general, for all $\mu \in [0,\lambda]$ and $r \in \mc X_v(\mc J, \zerobf, \mu)$. Moreover, whereas each recursion in \eqref{eq:centralized-routing} is a convex optimization problem (see Remark~\ref{rem:centralized-routing} (ii)), BPA does not enjoy this property in general. This is because, under \eqref{eq:Se-def}, the expression inside the minimization in \eqref{eq:S-tilde-def-new} is in general not affine in $x$, e.g., see Figure~\ref{fig:Rv0} (c) for an illustration. 
\end{enumerate}
\end{remark}

\subsection{BPA routing and lower bound on the margin of resilience}
In this section, we develop lower bounds for $\mor^*(\mc N,\lambda)$.
This will be done by analyzing a specific distributed routing policy, called \emph{BPA-routing}, whose construction is inspired by the Backward Propagation Algorithm. 
BPA routing is a routing policy that satisfies the following for all $v \in \mc V \setminus \{n\}$, $\mu \geq 0$: 
\be
\label{eq:BPA-routing}
\begin{split}
\rlb^*:=G^v(\mc E_v^+,\mu) & \in g\left(\mc E_v^+, \zerobf, \mu \right), \\
G^v\left(\mc J,\mu \right) & \in g\left(\mc J, \rlb^*, \mu \right), \quad \mc J \subset \mc E_v^+.
%
%
%\argmax_{x \in \real_+^{\mc J}: \, r \preceq x; \, \onebf.x=\mu} \, \, 
%\min_{e \in \mc J} \Big( S(e, r, x_e) + S\left(\mc J \setminus \{e\}, x, \mu\right) \Big), 
\end{split}
\ee
BPA routing derives its name from the fact that it relies on the function $g(\mc J, \rlb, \mu)$ from \eqref{eq:g-def-1}, which is directly related to the central computation in the BPA. However, note that the lower bound $\rlb^*$ in \eqref{eq:BPA-routing} is independent of $\mc J$ and $\mu$, and is always equal to the action of the routing policy under the same inflow $\mu$, when all local links are active, and with no lower bound constraint. Following \eqref{eq:g-def-implication}, $G^v(\mc J, \mu) \geq G^v(\mc E_v^+,\mu)$ for all $\mc J \subseteq \mc E_v^+$ and $\mu \geq 0$.
The following lemma formally states conditions under which BPA routing satisfies the link monotonicity in \eqref{eq:link-monotonicity}. 

\begin{lemma}
Let $\mc N$ be a network satisfying Assumption~\ref{ass:acyclicity} with $|\mc E_v^+| \leq 3$ for all $v \in \mc V \setminus \{n\}$. Then BPA routing defined in \eqref{eq:BPA-routing} and \eqref{eq:g-def-1} satisfies \eqref{eq:link-monotonicity}, and hence is a distributed routing policy as per Definition~\ref{def:oblivious}.
\end{lemma}
\begin{proof}
The only non-trivial case to prove is that, for every $j \in \mc J \subseteq \mc E_v^+$, $\mu \geq 0$, BPA routing satisfies: 
$$G^v(\mc E_v^+, \mu) \leq G^v(\mc E_v^+ \setminus \{j\}, \mu).$$
This is straightforward since \eqref{eq:g-def-implication} implies $G^v(\mc E_v^+ \setminus \{j\}, \mu) \geq \rlb^* = G^v(\mc E_v^+,\mu)$. 
\end{proof}

In general, BPA routing is not readily maximally resilient for general networks which are not \emph{directed trees} \footnote{Recall that $(\mc V, \mc E)$ is a directed tree if the undirected graph underlying $(\mc V, \mc E)$ is a tree.}. 
This is because if a node $v$ has multiple incoming links, then inactivation of $v$ results in inactivation of all the incoming links. However, 
the BPA algorithm does not take into account such correlations between link inactivations and hence the upper bound in Theorem~\ref{thm:upper-bound} is conservative for networks which are not trees. While it is possible to modify BPA algorithm to reduce this conservatism, this comes with additional computational complexity and additional difficulty in formulating the corresponding maximally resilient routing policy. Therefore, we make the following directed tree assumption in this paper for deriving lower bound on the margin of resilience.  

\begin{assumption}
\label{ass:polytree}
$\left(\mc V \setminus \{n\}, \mc E \setminus \mc E_n^- \right)$ is a directed tree.
\end{assumption}

With a slight abuse of terminology, we refer to $\mc N$ satisfying Assumption~\ref{ass:polytree} as a tree. Note that, $\mc N$ satisfying Assumption~\ref{ass:polytree} is a tree rooted at the unique origin node.

\begin{remark}
For a network satisfying Assumption~\ref{ass:polytree}, if $\lambda$ is less than the min cut capacity, then $f(0)$ under BPA routing is an equilibrium flow. Recall that the max flow min cut theorem implies that this is also a necessary condition for the existence of an equilibrium flow.   
\end{remark}

BPA routing is maximally resilient on flow networks which are trees and \emph{symmetric}. Recall that 
 a weighted rooted tree of depth
% \footnote{The depth of a directed tree is the length of the longest directed path in the tree.} 
 one is called symmetric if all the links outgoing from the root node have equal weights. 
A weighted rooted tree of depth greater than one is called symmetric if all the sub-trees rooted at the children\footnote{In a directed tree $(\mc V,\mc E)$, $u \in \mc V$ is called a children node of $v \in \mc V$ if $\mc E_u^- \cap \mc E_v^+ \neq \emptyset$.} nodes are symmetric, and identical to each other.

\begin{proposition}
\label{prop:symmetric}
Let $\mc N$ be a symmetric network satisfying Assumption~\ref{ass:polytree} with $\lambda  > 0$ a constant inflow at the origin node and BPA routing policy. Then, the associated network flow dynamics \eqref{eq:E-dynamics}-\eqref{eq:C-evolution}
is transferring for every disturbance process $(\delta(t))_{t\ge1}$ with $\Dnorm(\delta) < S^*(\mc N, \lambda)$.
\end{proposition}

The tree assumption is not sufficient for BPA routing to match the upper bound $S^*(\mc N, \lambda)$ given by the BPA for networks which are not symmetric, as illustrated in the following example. 

\begin{example}
\label{ex:flow-monotonicity-motivation}
%[On tightness of sufficient conditions in Theorem~\ref{thm:monotonicity-sufficient-conditions-graphical}]
Consider the graph topology from Figure~\ref{fig:depth4}, with $\lambda=2$, 
$\flowmax_{e_1}=2.5$, 
$\flowmax_{e_i}=3$ for $i=2,3$, $\flowmax_{e_i}=2$ for $i=4,7$, $\flowmax_{e_i}=0.6$ for $i=5,6$, $\flowmax_{e_8}=0.75$, $\flowmax_{e_9}=1.5$ and $\flowmax_{e_{10}}=0.17$. 
\begin{figure}[htb!]
\begin{center}
%\vspace{0.1in}
\includegraphics[width=8cm]{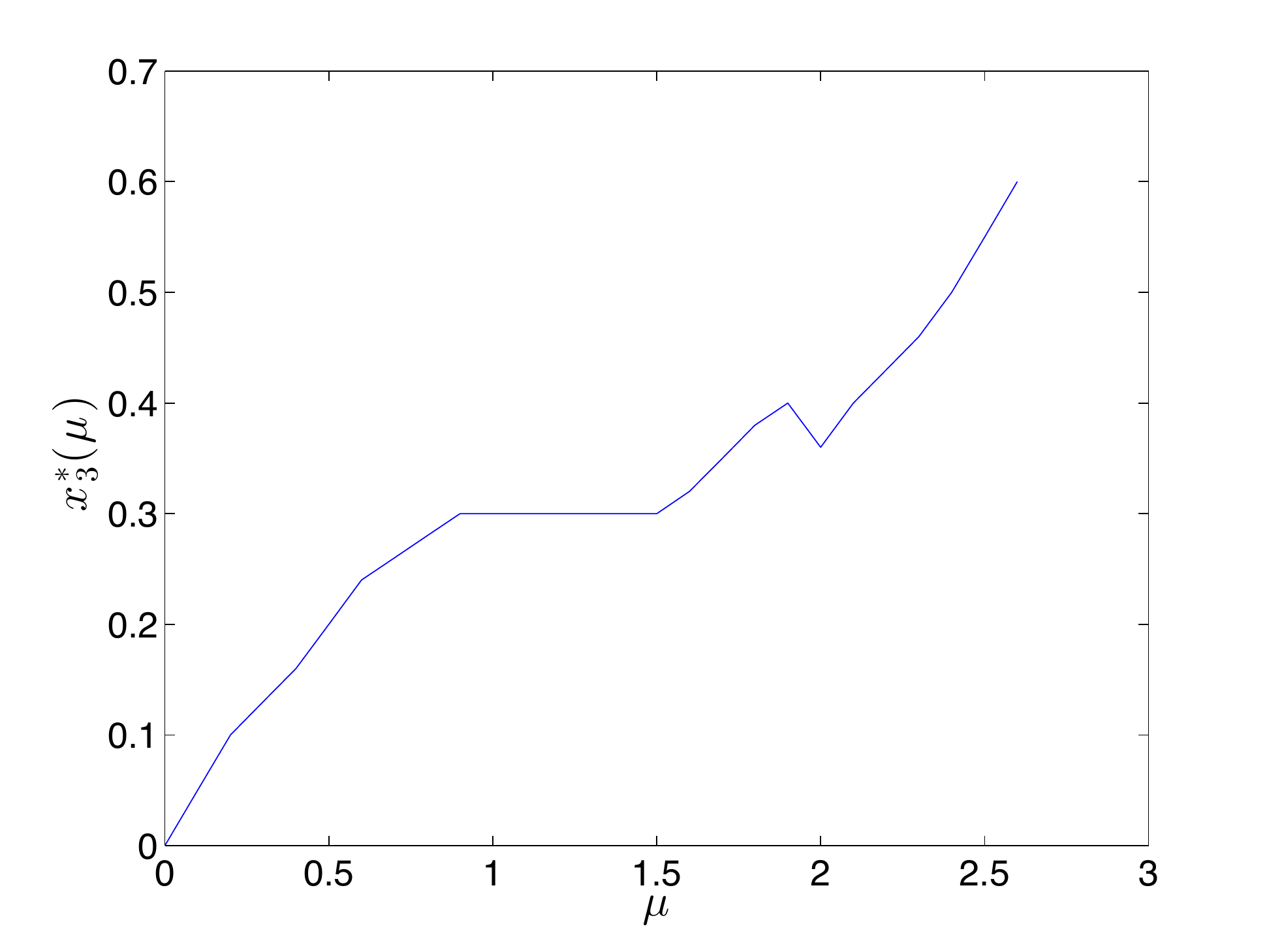}
\end{center}
\caption{Plot of $x_3^*(\mu):=G_{e_3}(\mc E_1^+(0),\mu)$ vs. $\mu$. }
\label{fig:nonmonotone}
\end{figure}
The plot of $x_3^*(\mu):=G_{e_3}(\mc E_1^+(0),\mu)$ vs. $\mu$ under BPA routing for these values is given in Figure~\ref{fig:nonmonotone}, which shows that $x_3^*(\mu)$ is decreasing in $\mu$ over $[1.9,2]$.
%, i.e., BPA routing at node $4$ is not flow monotone. 
Also, for these values, $S^*(\mc N,\lambda)=0.3$. Consider a disturbance process such that $\delta_5(1)=0.2$, $\delta_{10}(1)=0.07$, $\delta_i(1)=0$ for $i \in \until{15} \setminus \{2,4,5\}$ and $\delta(t) \equiv \zerobf$ for all $t \geq 2$. The magnitude of such a disturbance process is $0.27$, which is strictly less than $S^*(\mc N,\lambda)=0.3$. We now describe how such a disturbance process makes the associated network flow dynamics \eqref{eq:E-dynamics}-\eqref{eq:C-evolution} not transferring. 

Under BPA routing, $f(0)$ is such that: $2-f_2(0)=f_1(0)=1.9$.  Figure~\ref{fig:nonmonotone} then implies that $1.9-f_{4}(0)=f_3(0)=f_5(0)=0.4$. Therefore, under the given disturbance process, $\{e_{10}, e_5\}\notin \mc E(2)$, and $\{e_2, e_3\} \notin \mc E(3)$. Hence $f_{1}(4)=2$ and $f_4(5)=2=\flowmax_{e_4}$.  This implies that $e_4 \notin \mc E(6)$, and hence $e_1 \notin \mc E(8)$, which leads to the dynamics being not-transferring. 

On the other hand, it is easy to see that the dynamics would be transferring under this disturbance process if the routing policy at node $1$ is such that $f_3(0) < 0.4$, and $f_3(0)=x_3^*(2)=0.35$ (see Figure~\ref{fig:nonmonotone}) in particular. This would correspond to the routing policy at node $1$ anticipating its inflow in advance, which is not feasible under the oblivious and distributed setting for routing policies.     
%$\chi_{e_5}(0)=\xi_{e_5}(0)=\psi_{3}(0)=0$ and $\chi_{e_{10}}(0)=\xi_{e_{10}}(0)=\psi_2(0)=0$. Consequently, $\chi_{e_3}(1)=\xi_{e_3}(1)=0$, and $\chi_{e_2}(1)=\xi_{e_2}(1)=0$. Therefore, $f_{e_4}(2)=1.9$ and $f_{e_1}(2)=2$. Consequently, $f_{e_4}(3)=2>\flowmax_{e_4}-\delta_4=2-\epsilon$. Therefore, $\chi_{e_4}(3)=\xi_{e_4}(3)=\psi_1(3)=0$, and hence $\chi_{e_1}(4)=\xi_{e_1}(4)=\psi_0(4)=0$, which implies that the network is not transferring.
%Since $|\mc E_v^+| \leq 2$ for all $v \in \mc V \setminus \{n\}$ for the graph topology in Figure~\ref{fig:nonmonotone}, part (ii) of Remark~\ref{rem:upper-bound} implies that $\mc R^*(\lambda,\mc N) \leq R_0(\lambda)$, which is consistent with the findings of this example. 
\end{example}

Example~\ref{ex:flow-monotonicity-motivation} suggests that the non-monotonicity in the control action of BPA routing, and hence in the evolution of flows on the links,  under \emph{point-wise} (with respect to inflow) optimization could lead to its sub optimality. This motivates consideration of the following additional constraint.

\begin{definition}
%%[Flow-monotone Oblivious Routing Policy]
\label{def:flow-monotone-routing}
A distributed routing policy $\mc G$ is called flow-monotone at node $v \in \mc V \setminus \{n\}$ if, for every $\mc J \subseteq \mc E_v^+$: 
\begin{equation}
\label{eq:flow-monotonicity}
0 \leq \mu_1 \leq \mu_2 \implies G^v(\mc J,\mu_1) \leq G^v(\mc J,\mu_2),
\end{equation}
%\kscomment{where $\BPAflowmax_{e}:=\sup\setdef{\mu \geq 0}{S_{e}(\mu) > 0}$ is the effective capacity of link $e$.}
\end{definition}
Under a flow-monotone routing policy, if the inflow at a node increases, then the 
flow assigned to every active outgoing link from that node does not decrease.
A routing policy which is flow monotone over all $v \in \mc V \setminus \{0,n\}$, is said to be flow monotone over $\mc N$. 
We exclude the origin node because the inflow $\lambda$ at the origin node is fixed. 

\begin{remark}
\label{rem:link-monotonicity-justification}
Note that, unlike the link monotonicity condition in \eqref{eq:link-monotonicity}, we did not include the flow monotonicity condition in \eqref{eq:flow-monotonicity} as part of the definition of distributed routing policies. This is because, while Example~\ref{ex:flow-monotonicity-motivation} illustrates that BPA routing is not necessarily flow monotone, we have not been able to find an example where link monotonicity is violated by BPA routing with $\rlb^*=\zerobf$ in \eqref{eq:BPA-routing}. However, a mathematical proof to support this observation is lacking at this point.
\end{remark}

Under a flow monotone distributed routing policy, the network dynamics can be easily shown to possess the following simple property (which we state without proof), which simplifies the analysis considerably.

\begin{lemma}
\label{lem:monotonicity}
Let $\mc N$ be a network with $|\mc E_v^+| \leq 3$ for all $v \in \mc V \setminus \{n\}$ and satisfying Assumption~\ref{ass:polytree}, $\lambda  > 0$ a constant intflow at the origin node and BPA routing policy that is flow monotone. Then,  
$$t_1 \leq t_2 \implies f_e(t_1) \leq f_e(t_2) \qquad \forall \, e \in \mc E(t_2).$$
\end{lemma}

%\begin{remark}
%\label{remark:monotonicity}
%Under a distributed routing policy that is flow monotone, the evolution of network flow possesses the following monotone property, which simplifies the analysis considerably: $\forall \, 0 \leq t_1 \leq t_2$ and $e \in \mc E(t_2) \implies f_e(t_1) \leq f_e(t_2)$. 
%\end{remark}

%\kscommentphantom{
%The following result characterizes some specific classes of flow networks for which flow monotonicity guarantees that BPA routing is also maximally resilient.
%\ksmargin{should we keep Proposition~\ref{prop:flow-monotone-sufficient} ?}
%\begin{proposition}
%\label{prop:flow-monotone-sufficient}
%Consider a flow network $\mc N$, $\lambda  > 0$ a constant total outflow at the origin node, and operating under BPA routing policy which is flow monotone. Then, the network is transferring for every disturbance process $\delta$ with $\Dnorm(\delta) < R_{0}(\lambda)$ if $\mc N$ satisfies any of the following two conditions in addition to Assumptions~\ref{ass:connectivity} and \ref{ass:polytree}:
%\begin{enumerate}
%\item $|\mc E_v^+| \leq 2$ for every $v \in \mc V \setminus \{n\}$ that is not a leaf node; or 
%\item  $|\mc E_0^+| \leq 3$ and depth of $\mc N$ is two.
%\end{enumerate}
%\end{proposition}
%}

The following is a key result, which, along with Theorem~\ref{thm:upper-bound},  identifies conditions under which BPA routing is maximally resilient.

\begin{theorem}
\label{thm:lower-bound}
Let $\mc N$ be a network with $|\mc E_v^+| \leq 3$ for all $v \in \mc V \setminus \{n\}$ and satisfying Assumption~\ref{ass:polytree}, $\lambda  > 0$ a constant inflow at the origin node and BPA routing policy that is flow monotone. Then, the associated network flow dynamics \eqref{eq:E-dynamics}-\eqref{eq:C-evolution}
is transferring for every disturbance process $(\delta(t))_{t\ge1}$ with $\Dnorm(\delta) < S^*(\mc N, \lambda)$.
\end{theorem}

Since BPA routing is completely specified by network parameters $(\mc V, \mc E, \flowmax)$, flow monotonicity is a condition on the network parameters.
BPA routing is flow monotone at $v$ trivially if $|\mc E_v^+|=1$. 
One could perform extensive (offline) numerical tests to check flow-monotonicity of BPA routing over a given flow network $\mc N$. 
However, it is possible to identify a few flow networks over which BPA routing is provably flow-monotone. 
In order to characterize such networks in Proposition~\ref{prop:monotonicity-sufficient-conditions-graphical} and Remark~\ref{remark:flow-monotone-construction}, we need the concept of \emph{d-expansion} of a network: given $\mc N$, its d-expanded version $\mc N^d$ is obtained by creating multiple copies of the destination node in $\mc N$, one for each incoming link. For example, the network in Figure~\ref{fig:graphical-sufficient-conditions}(a) is d-expanded version of the network in Figure~\ref{fig:simple-setting} (a). It is easy to recover the original flow network from its d-expanded version.

%In order to characterize such networks, we need a few more definitions. 

%\kscomment{(concise definition of \emph{cascade relevant sub-tree} comes here)}

\begin{figure}[htb!]
\begin{center}
%\begin{minipage}[c]{.175\textwidth}
%\begin{center}
%\includegraphics[width=1.5cm]{./fig/casea}
%\end{center}
%\end{minipage}
\begin{minipage}[c]{.175\textwidth}
\begin{center}
\includegraphics[width=2cm]{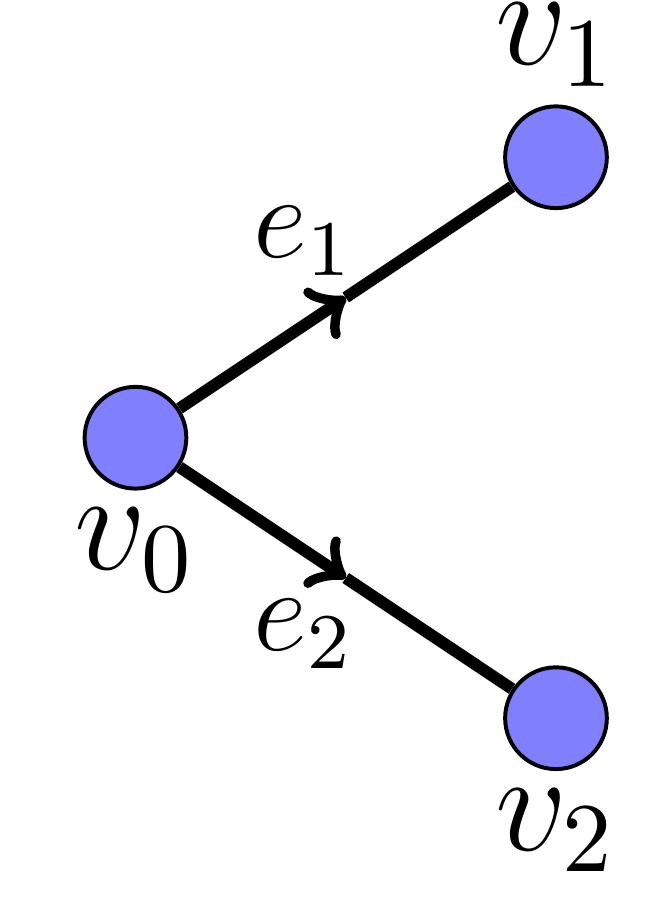} 
\end{center}
\end{minipage}
%\begin{minipage}[c]{.175\textwidth}
%\begin{center}
%\includegraphics[width=2.5cm]{./fig/cased} 
%\end{center}
%\end{minipage}
\begin{minipage}[c]{.2\textwidth}
\begin{center}
\includegraphics[width=3cm]{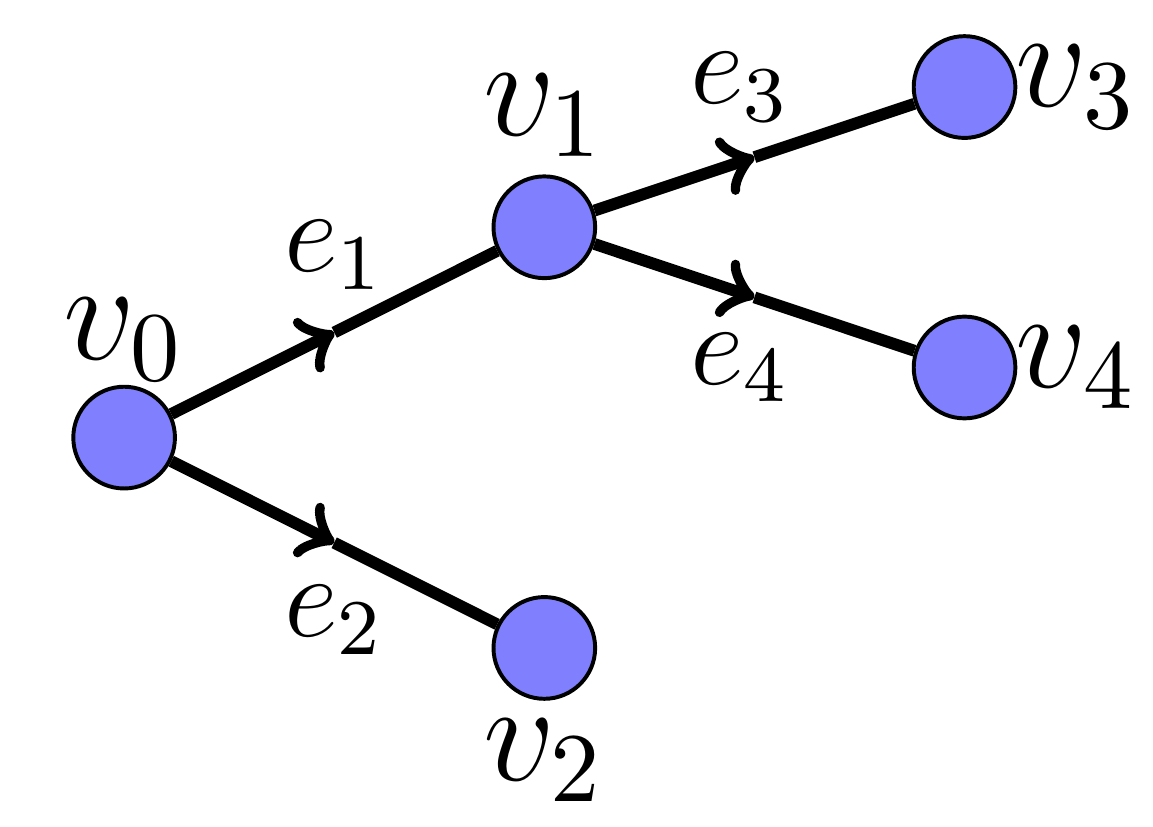} 
\end{center}
\end{minipage} 
\begin{minipage}[c]{.2\textwidth}
\begin{center}
\includegraphics[width=3cm]{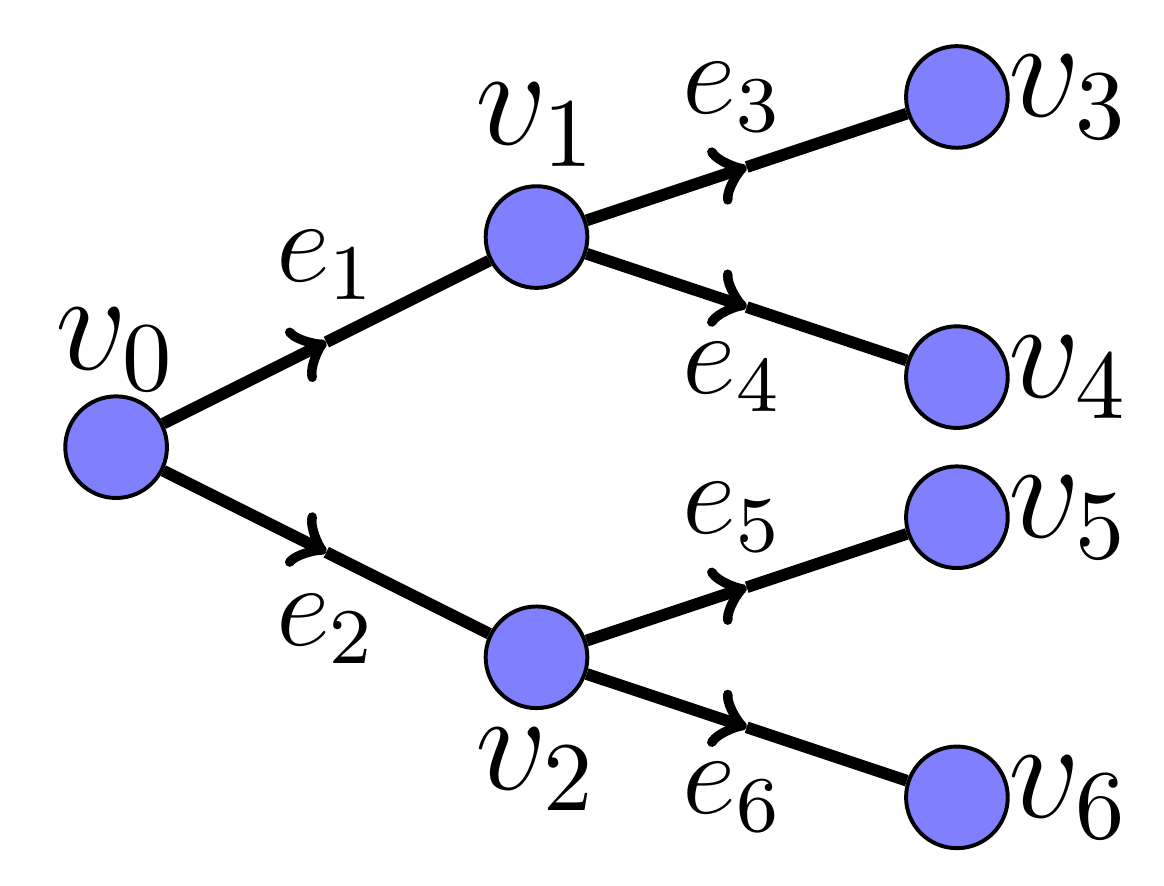} 
\end{center}
\end{minipage}

%\begin{minipage}[c]{.175\textwidth}
%\begin{center}
%(a)
%\end{center}
%\end{minipage}
\begin{minipage}[c]{.175\textwidth}
\begin{center}
(a)
\end{center}
\end{minipage}
%\begin{minipage}[c]{.175\textwidth}
%\begin{center}
%(b)
%\end{center}
%\end{minipage}
\begin{minipage}[c]{.2\textwidth}
\begin{center}
(b)
\end{center}
\end{minipage} 
\begin{minipage}[c]{.2\textwidth}
\begin{center}
(c)
\end{center}
\end{minipage}
\end{center}
\caption{d-expanded flow networks that induce flow-monotonicity of BPA routing at the root node $v_0$.}
\label{fig:graphical-sufficient-conditions}
\end{figure}

%\begin{proposition}
%\label{prop:monotonicity-sufficient-conditions-graphical}
%The BPA routing is flow-monotone if, in addition to satisfying Assumptions~\ref{ass:connectivity} and \ref{ass:polytree}, $\mc N$ is such that the cascade-relevant sub-tree of every node in meta-tree $\mc N^m$ is either (a), or (b), or (c), or (d) with $\flowmax_{e_2} \geq \flowmax_{e_1}$ or (e) with $\flowmax_{e_1}=\flowmax_{e_2}$ in Figure~\ref{fig:graphical-sufficient-conditions}.
%\end{proposition}

\begin{proposition}
\label{prop:monotonicity-sufficient-conditions-graphical}
Let $\mc N$ be a network satisfying Assumption~\ref{ass:polytree} with $\lambda  > 0$ a constant inflow at the origin node. 
%Let $\mc V_0:=\setdef{v \in \mc V \setminus \{0,n\}}{(0,v) \in \mc E}$ be the set of nodes immediately downstream of the origin node. 
Then, BPA routing is flow-monotone at $v \in \mc V$ if the the sub-tree in the d-expanded version $\mc N^d$ rooted at $v$ is either (a), or (b) with $\flowmax_{e_2} \geq \flowmax_{e_1}$ or (c) with $\flowmax_{e_1}=\flowmax_{e_2}$ in Figure~\ref{fig:graphical-sufficient-conditions}.
\end{proposition}

\begin{remark}
\label{remark:flow-monotone-construction}
Let $\bar{\mc N}^d$ denote the set of d-expanded versions of ``\emph{simple}" flow networks over which BPA routing is known to be flow monotone, either from explicit analysis as in Proposition~\ref{prop:monotonicity-sufficient-conditions-graphical}, or through extensive simulations. One can use $\bar{\mc N}^d$ as a basis to form arbitrarily large networks over which BPA routing is flow monotone, using an iterative procedure as follows. Initialize $\mc N^d$ to be an element of $\bar{\mc N}^d$. At each iteration, execute the following concatenation step. Take any  destination node in $\mc N^d$, say $v$, with a single incoming link $\mc E_v^- =\{e\}$ and a member $\bar{\mc N}^d_i$ from $\bar{\mc N}^d$ whose origin node is $v_{0,i}$. If $[\flowmax_e-\mu]^+ \leq S(\mc E_{v_{0,i}}^+, \zerobf, \mu)$ for all $\mu \in [0,\flowmax_e]$, then concatenate $\bar{\mc N}^d_i$ to $\mc N^d$ at $v$, i.e., $v=v_{0,i}$ and the the sub-network downstream of $v$ is $\bar{\mc N}^d_i$. 
BPA routing on the flow network obtained at the end of every iteration of this procedure is guaranteed to be flow monotone, because \eqref{eq:Se-def} implies that $S_e(\mu)$ is equal to $[\flowmax_e-\mu]^+$ even in the concatenated network. 

Flow monotonicity of BPA routing over a given network $\mc N$ is maintained even after replacing any link, say $e=(u,v)$, in $\mc N^d$ (at any iteration) with a non-branching chain $e_1, \ldots, e_k$ such that $\sigma_{e_1}=u$, $\tau_{e_m}=v$ and $\min_{i \in \until{m}} \flowmax_{e_i}=\flowmax_e$. 

One can also devise a procedure which is counterpart to the \emph{expansion} procedure described above 
to check if the d-expanded version of a given network $\mc N^d$ can be decomposed into elements of $\bar{\mc N}^d$, in which case the BPA routing over $\mc N^d$ is flow monotone. 
As one keeps enriching the basis $\bar{\mc N}^d$, these procedures allow to construct or to verify large networks over which BPA routing is flow monotone, and hence maximally resilient by Theorem~\ref{thm:lower-bound}. 
\end{remark}

%While Proposition~\ref{prop:flow-monotone-sufficient} does not cover all the possible polytree flow networks, our preliminary numerical studies have not resulted in any polytree flow network over which BPA routing is flow monotone, but not maximally resilient. Nevertheless, it is possible to state a general sufficient condition in addition to flow monotonicity, under which BPA routing is maximally resilient over any polytree flow network. 
%In order to state this sufficient condition, we need the notion of \emph{link monotoncity} of a routing policy, which is formally defined as follows. 

\section{Proofs of the Main Results}
\label{sec:proofs}
In this section, we provide proofs of the main results presented in Section~\ref{sec:main-results}. Some of the proofs rely on certain analytical properties of $S(\mc J,\rlb,\mu)$ defined in \eqref{eq:S-tilde-def-new}. We state and prove these properties in Lemma~\ref{lem:s-function-decreasing} in the Appendix.

\subsection{Proof of Proposition~\ref{prop:centralized-upper-bound}}
\label{subsec:centralized-ub-proof}
We first show that, for any $\lambda \geq 0$, $\mc K \subseteq \mc J \subseteq \mc E$,
\be
\label{eq:K-J-ineq}
S(\mc K, \lambda) \leq S(\mc J, \lambda).
\ee
%It is sufficient to show that $S(\mc J \setminus \{e\}, \lambda) \leq S(\mc J, \lambda)$ for all $e \in \mc J$. 
%This follows from \eqref{eq:centralized-routing}, which implies that there exists $x \in \mc X(\mc J, \lambda)$ such that $S(\mc J, \lambda) = \flowmax_e - x_e + S(\mc J \setminus \{e\}, \lambda) \geq S(\mc J \setminus \{e\}, \lambda)$ for every $e \in \mc J$. 
It suffices to show that $S(\mc J \setminus \{e\}, \lambda) \leq S(\mc J, \lambda)$ for all $e \in \mc J$. It is trivially true for $|\mc J|=1$. Assume it to be true for all $|\mc J| \leq k$ for some $k \geq 1$. \eqref{eq:centralized-routing} implies that, for all $x \in \mc X(\mc J,\lambda)$: 
\begin{align}
\nonumber 
S(\mc J, \lambda) & \geq \min_{j \in \mc J} \big(\flowmax_j - x_j + S(\mc J \setminus \{j\}, \lambda) \big) \\ 
\nonumber
& = \min \left (C_e-x_e + S(\mc J \setminus \{e\},\lambda), \min_{j \in \mc J \setminus \{e\}} \big(\flowmax_j - x_j + S(\mc J \setminus \{j\}, \lambda) \big) \right) \\
\nonumber 
& \geq \min \left (S(\mc J \setminus \{e\},\lambda), \min_{j \in \mc J \setminus \{e\}} \big(\flowmax_j - x_j + S(\mc J \setminus \{j\}, \lambda) \big) \right) \\
\label{eq:SJ-longeq} 
& \geq \min \left (S(\mc J \setminus \{e\},\lambda), \min_{j \in \mc J \setminus \{e\}} \big(\flowmax_j - x_j + S(\mc J \setminus \{j,e\}, \lambda) \big) \right),
\end{align}
where the second inequality follows the fact that $x_e \leq \flowmax_e$, whereas the third inequality follows from the inductive argument on $\mc J \setminus \{i\}$.
For $x \in \argmax_{x \in \mc X(\mc J \setminus \{e\},\lambda)} \min_{j \in \mc J \setminus \{e\}}  \big(\flowmax_j - x_j + S(\mc J \setminus \{j,e\}, \lambda) \big)$ 
we get $S(\mc J \setminus \{e\}, \lambda) = \min_{j \in \mc J \setminus \{e\}} \big(\flowmax_j - x_j + S(\mc J \setminus \{j,e\}, \lambda)$, which when used in \eqref{eq:SJ-longeq}, finishes the proof for $|\mc J|=k+1$. By induction, \eqref{eq:K-J-ineq} is then true for all $\mc J \subseteq \mc E$. 

We now prove the proposition by induction on $|\mc J|$. When $\mc J = \{e\}$, \eqref{eq:centralized-routing} implies that $S(\mc J, \lambda) \leq [\flowmax_e-\lambda]^+$. It is easy to see that under a disturbance process $\delta_e(1)=[\flowmax_e-\lambda]^+$ and $\delta_e(t)=0$ for all $t \geq 1$, the associated network dynamics will be non-transferring. Assume that the proposition is true for all $\mc J \subseteq \mc E$ with $|\mc J| \leq k$ for some $k \geq 1$. Let $\mc E(0)=\mc J$,  with $|\mc J|=k+1$.
%Let $f(0) \in \mc X(\mc J, \lambda)$ be the initial flow vector under a given distributed routing policy. 
Pick $e \in \argmin_{j \in \mc J} \left(\flowmax_j - f_j(0) + S(\mc J \setminus \{j\}, \lambda) \right)$. Therefore,  
\be
\label{eq:SJ-ineq}
S(\mc J, \lambda) \geq \flowmax_e - f_e(0) + S(\mc J \setminus \{e\}, \lambda).
\ee

Consider a disturbance process such that $\delta_{e}(1)=\flowmax_e-f_e(0)$ and $\delta_j(1)=0$ for all $j \in \mc J \setminus \{e\}$. Under this disturbance, link $e$ becomes inactive, followed by a possible cascading failure. Let the network state come to a steady state after a finite time $\mc T_1$. Since, $|\mc E(\mc T_1)| \leq k$, one can use induction to extend $\delta$ after $\mc T_1$ to ensure that the network dynamics is 
not transferring. By induction, the total magnitude of $\delta$ is then upper bounded as 
$\Dnorm(\delta) \leq \delta_{e}(1) + S(\mc E(\mc T_1),\lambda)$. Since $\mc E(\mc T_1) \subsetneq \mc J$, using \eqref{eq:K-J-ineq}, this can be further upper bounded as $\Dnorm(\delta) \leq \delta_{e}(1) + S(\mc J \setminus \{e\},\lambda)=\flowmax_e-f_e(0) + S(\mc J \setminus \{e\},\lambda)$, which combined with \eqref{eq:SJ-ineq} implies that 
$\Dnorm(\delta) \leq S(\mc J, \lambda)$. 

\subsection{Proof of Proposition~\ref{prop:transferring-notion-equivalence}}
If $0 \notin \mc V(\mc T)$, then the flow across links outgoing from every cut $\U$ in $(\mc V(\mc T), \mc E(\mc T))$, and $\V(\mc T) \setminus \{n\}$ in particular, is zero at $\mc T$. This implies that $\lambda_n(\mc T)=0$. This also proves that $\lambda_n(\mc T) \in \{0,\lambda\}$. 

If $0 \in \mc V(\mc T)$, then $\mc J:= \mc E(\mc T) \cap \mc E_0^+$ is non-empty, and $\sum_{e \in \mc J}f_e(\mc T)=\lambda$. Let $\mc U=0 \cup \{v \in \mc V(\mc T) : (0,v) \in \mc J\}$. It is clear that $\mc U \subseteq \mc V(\mc T)$, and that $\mc U$ is in fact a cut in $(\mc V(\mc T), \mc E(\mc T))$. Since $f(\mc T)$ is an equilibrium flow, 
the total flow across the links outgoing from $\mc U$ is $\lambda$. 
One can continue along these lines to claim that the flow across links outgoing from any cut in $(\mc V(\mc T), \mc E(\mc T))$, and $\V(\mc T) \setminus \{n\}$ in particular, is equal to $\lambda$. That is, $\lambda_n(\mc T)=\lambda$.

\subsection{Proof of Theorem~\ref{thm:upper-bound}} 
\label{subsec:ub-proof}
%\kscomment
%{
%\begin{lemma}
%For any $\rlb \in \mc X_v(\mc E_v^+, \zerobf,\mu) \cup \{\zerobf\}$, $\emptyset \neq \mc K \subseteq \mc J \subseteq \mc E_v^+$, $v \in \mc V \setminus \{n\}$, $\mu \geq 0$:
%$$\mc X_v(\mc K, \rlb,\mu) = \emptyset \, \, \iff \, \, \mc X_v(\mc J, \rlb,\mu)=\emptyset.$$
%\end{lemma}
%\begin{proof}
%$\mc K \subseteq \mc J$ implies $\mc X_v(\mc K, \rlb,\mu) \subseteq \mc X_v(\mc J, \rlb,\mu)$, which gives $(\Longleftarrow)$. 
%In order to prove $(\Longrightarrow)$, assume by contradiction that $\mc X_v(\mc J, \rlb_1,\mu) = \emptyset$ but $x \in \mc X_v(\mc J, \rlb_2,\mu)$. 
%\end{proof}
%}
Theorem~\ref{thm:upper-bound} is a corollary of the following lemma, where we allow the possibility that $(\mc V(0),\mc E(0)) \neq (\mc V, \mc E)$. 

%In the statement of the following lemma, $R_v(\mu)$ refers to the quantity computed by BPA in Algorithm~\ref{algo:BPA}. 
%We need the following concepts for the lemma. For $v \in \mc V$, let $\mc U^v:= \{v\} \cup \setdef{u \in \mc V}{\exists \text{ a path from }v \text{ to }u \text{ in } \mc T}$ be the set of nodes reachable from $v$, and $\mc E^{v}:=\{(u,w) \in \mc E, \, u,w \in \mc U^v\}$ be the set of all links that connect the nodes in $\mc U^v$. Similarly, let $\flowmax^v$ be the restriction of $\flowmax$ to the links in $\mc E^{v}$. In summary, $\mc N^v:=(\mc U^v, \mc E^v,C^v)$ represents the flow sub-network reachable from $v$. 

\begin{lemma}
\label{lem:upperbound-general-lemma}
Consider a node $v$ in a network $\mc N$ with initial condition $(\mc V(0),\mc E(0))\subseteq (\mc V, \mc E)$, satisfying Assumption~\ref{ass:acyclicity}, with a constant inflow $\mu \geq 0$, and operating under a 
distributed routing policy satisfying $G^v(\mc E_v^+(0), \mu) \geq \rlb$ for some  
$\rlb \in \real_+^{\mc E_v^+}$. Then, for any $h \in \natural$, there exists a finite $\mc T_v \geq h$, and a disturbance process $(\delta^{v}(t))_{t\ge h}$ satisfying $\Dnorm(\delta^{v}) \leq S(\mc E_v^+(0),\rlb,\mu)$, under which
$v \notin \mc V(\mc T_v).$
\end{lemma}
\begin{proof}
%\ksmargin{what about $\xi(0) \neq \onebf$ ?}
%We prove through a double induction: on the position of $v$ with respect to the topological ordering $\zerountil{n}$, and on the cardinality of $\mc E_v^+$. 
It is sufficient to prove the lemma for $h=1$. 
%Let $\mc E_v^+(0)=\{e_1, \ldots, e_m\}$. 
For brevity in notation, we let $f(t):=G^v(\mc E_v^+(t),\mu)$ be the action of the control policy at $t \geq 0$. By assumption and the link monotonicity property of routing policy in \eqref{eq:link-monotonicity}, $f(t_2) \geq f(t_1) \geq \rlb$ for all $t_2 \geq t_1 \geq 0$. 
%Also, $\rlb \in \mc X_v(\mc E_v^+,\zerobf,\mu) \cup \{\zerobf\}$ implies $r_i \leq \flowmax_i$ for all $i \in \until{m}$.
%\kscomment{We also let 
%\be
%\label{eq:Ev-def}
%\mc E^v:=\setdef{e \in \mc E}{\exists \text{ a directed path in } (\mc V, \mc E) \text{ from }v \text{ to } \sigma_e \text{ and } \tau_e}.
%\ee
%}
We follow the convention that, unless specified otherwise, $\delta^v_e(t)=0$ for all $e \in \mc E$ and $t \geq 1$. 
%\kscomment{We follow the same convention for an another disturbance process $\tdelta^{\mc J,h}$ used in the proof.} 

The proof is by double induction, on the number of nodes $n+1$ and the cardinality of $|\mc E_v^+(0)|$. The proof is easy to verify for acyclic networks with $n=1$ node and $|\mc E_v^+(0)|=1$, since in that case, with $\mc E_v^+(0)=\{e\}$, $S_{e}(\mu)  =[\flowmax_{e} - \mu]^+$, and therefore, one can apply 
$\delta^v_{e}(1)=[\flowmax_e-\mu]^+$, under which $v \notin \mc V(2)$ and $\Dnorm(\delta^v) = [\flowmax_e-\mu]^+$. 
%
%The proof is by induction on the topological ordering of $v$ in $\zerountil{n-1}$. First consider $v=n-1$, in which case $\tau_{e_i}=n$ for all $i \in \until{m}$. If $m=1$, then $r_{1} \leq \flowmax_{1}$ and \eqref{eq:Se-def} imply $S(\mc E_v^+(0),\rlb,\mu)=[\flowmax_1-\mu]^+$ for all $\rlb \in \mc X_v(\mc E_v^+,\zerobf,\mu) \cup \{\zerobf\}$. Let $\delta^v_{e_1}(1)=[\flowmax_1-\mu]^+$, under which  $v \notin \mc V(2)$ and $\Dnorm(\delta^v) = [\flowmax_1-\mu]^+ = S(\mc E_v^+(0),\rlb,\mu)$. 
Assume the lemma to be true for arbitrary acyclic networks with $n+1$ nodes and $|\mc E_v^+(0)| \leq k$ for some $k \geq 1$. 

For $|\mc E_v^+(0)|=k+1$, pick $e$ in $\argmin_{j \in \mc E_v^+(0)} \big( [\flowmax_j-f_j(0)]^+ + S\left(\mc E_v^+(0) \setminus \{j\}, y, \mu\right) \big)$. Consider a disturbance process $\delta^v$ such that $\delta^v_e(1)=[\flowmax_e-f_e(0)]^+$, in which case 
\be
\label{eq:ub-eq1}
S(\mc E_v^+(0), \rlb, \mu) \geq [\flowmax_e-y_e]^+ + S(\mc E_v^+(0) \setminus \{e\},f(0),\mu).
\ee

Let the times at which links fail simultaneously be $2=t_1 \leq \ldots \leq t_m$. Let $S^t(\mc J, \rlb, \mu)$ denote the functions computed by the BPA in \eqref{eq:S-tilde-def-new} for the residual graph $(\mc V(t), \mc E(t), \flowmax(t))$ at $t \geq 0$. 
%sequence of links to become inactive simultaneously under $\delta^v$ be $(\mc K_1, \ldots, \mc K_m)$ at times $(t_1, \ldots, t_m)$, where $\mc K_i \subseteq \mc E_v^+(0)$ for all $i \in \until{m}$, $e \in \mc K_1$ and $t_1=2$. Let $S^{i}(\mc J, \rlb, \mu)$ denote the functions computed by the BPA in \eqref{eq:S-tilde-def-new} for the residual graph $(\mc V(t_i), \mc E(t_i), \flowmax(t_i))$ at times $t_i$, $i \in \until{m}$. Let $t_0:=0$, and $S^0 \equiv S$. We next prove that 
%$$S^i\left(\mc E_v^+(t_i), f(t_i), \mu) \right) \leq S^{i-1}\left(\mc E_v^+(t_{i-1}), \rlb, \mu\right) \qquad \forall i \in \until{m}.$$
%Consider an arbitrary $i \in \until{m}$. 
By convention, we set $S^0 \equiv S$. 
Since links in $\mc E_v^+(t_{i-1}) \setminus \mc E_v^+(t_{i})$ fail simultaneously at $t_i$, $S_j^{t_i-1}(f_j(t_i-1))=0$ for all $j \in \mc E_v^+(t_{i-1}) \setminus \mc E_v^+(t_{i})$. Therefore, using Lemma~\ref{lem:s-function-decreasing}, for all $i \in \{2, \ldots, m\}$:
\be
\label{eq:Se-telescopic}
\begin{split}
S^{t_i}(\mc E_v^+(t_i),f(t_i),\mu) & \leq S^{t_i}(\mc E_v^+(t_i),f(t_i-1),\mu) \leq S^{t_i-1}(\mc E_v^+(t_i),f(t_i-1),\mu) \\ & = S^{t_i-1}(\mc E_v^+(t_i-1),f(t_i-1),\mu) \leq S^{t_{i-1}}(\mc E_v^+(t_{i-1}),f(t_{i-1}),\mu),
\end{split}
\ee
where we have used the fact that $f(t_i) \geq f(t_i-1)=f(t_{i-1})$ and $\mc E(t_i-1) \subseteq \mc E(t_{i-1})$. 
Using the same arguments, since $S^1_e(f_1(1))=0$, we have that 
\be
\label{eq:Se-telescopic-base}
S^{t_1}\big(\mc E_v^+(t_1), f(t_1), \mu \big) \leq S^1\big(\mc E_v^+(1), f(1), \mu \big) = S^{1}\big(\mc E_v^+(0) \setminus \{e\}, f(1), \mu \big) \leq S\big(\mc E_v^+(0) \setminus \{e\}, \rlb, \mu \big).
\ee

Combining \eqref{eq:Se-telescopic} and \eqref{eq:Se-telescopic-base}, we get that 
\be
\label{eq:ub-eq2}
S^{t_m}\big(\mc E_v^+(t_m), f(t_m), \mu\big) \leq S\big(\mc E_v^+(0) \setminus \{e\}, \rlb, \mu\big).
\ee
Using induction on the residual graph at $t=t_m$, where $|\mc E_v^+(t_m)| < |\mc E_v^+(0)|$, one can construct a disturbance process $(\tdelta^{v}(t))_{t\ge t_m}$ such that $v \notin \mc V(\mc T)$, and 
\be
\label{eq:ub-eq3}
D(\tdelta^v) \leq S^{t_m}\big(\mc E_v^+(t_m), f(t_m), \mu\big).
\ee 
Augmenting $\delta^v$ with $\tdelta^v$, i.e., $(\delta^v(t))_{t \geq t_m}=(\tdelta^{v}(t))_{t\ge t_m}$, and applying \eqref{eq:ub-eq1}, \eqref{eq:ub-eq2} and \eqref{eq:ub-eq3}, we get that 
$$\Dnorm(\delta^v)=[\flowmax_e-f_e(0)]^+ +\Dnorm(\tdelta^{v}) \leq S(\mc E_v^+(0), f(0), \mu) \leq S(\mc E_v^+(0), \rlb, \mu),$$
which proves the lemma for acyclic networks with $n+1$ nodes and $|\mc E_v^+(0)| = k+1$. The proof can be easily extended to acyclic networks with $n+2$ nodes and $|\mc E_v^+(0)|=k$, after which the lemma follows from induction. 
\end{proof}

Theorem~\ref{thm:upper-bound} follows by combining Lemma~\ref{lem:upperbound-general-lemma} for $v=0$, $h=1$, $\mu=\lambda$, $\rlb = \zerobf$ and $(\mc V(0), \mc E(0))=(\mc V, \mc E)$ with Proposition~\ref{prop:transferring-notion-equivalence}.

%\subsection{Proof of Proposition~\ref{prop:depth-one}}
%\kscomment{(extract a concise version from the long draft here)}
%We need an extension of the concept in \eqref{eq:Ev-def} for this section. 
%For $\mc J \subseteq \mc E_v^+$, $v \in \mc V \setminus \{n\}$, let
%\be
%\label{eq:Ev-def-general}
%\mc E^v_{\mc J}:= \union_{e \in \mc J} \mc E^{\tau_e} \cup \mc J
%\ee
%be the set of links in $\mc E$ that are reachable from $v$ and pass through $\mc J$. 
%$\mc E^v_{\mc E_v^+}$ corresponds to $\mc E^v$ defined in \eqref{eq:Ev-def}.
\subsection{Proof of Proposition~\ref{prop:symmetric} and Theorem~\ref{thm:lower-bound}}
Proposition~\ref{prop:symmetric} follows from Theorem~\ref{thm:lower-bound} by recalling that, for symmetric flow networks that are directed trees, BPA routing satisfies flow monotonicity.
Theorem~\ref{thm:lower-bound} follows from the following lemmas.
The following simple property of the functions $S(\mc J,\rlb,\mu)$ computed in \eqref{eq:S-tilde-def-new} will be useful in the proofs in this subsection.
We recall the definition of $g(\mc J,\rlb,\mu)$ from \eqref{eq:g-def-1}.
\begin{lemma}
\label{lem:S-split-upperbound}
For any $x \in g(\mc J, \rlb, \mu)$, $\mc J \subseteq \mc E_v^+$, $\rlb \in \real_+^{\mc E_v^+}$, $v \in \mc V \setminus \{n\}$ and $\mu \geq 0$, 
$$S(\mc J, \rlb, \mu) \leq \sum_{e \in \mc K} S_e(x_e) + S(\mc J \setminus \mc K, x, \mu) \qquad \forall \mc K \subseteq \mc J.$$
\end{lemma}
\begin{proof}
The lemma is trivially true from \eqref{eq:S-tilde-def-new} and \eqref{eq:g-def-1} for $|\mc K|=1$. Assume it to be true for all $\mc K \subset \mc J$ with $|\mc K| \leq k$ for some $k \geq 1$. Consider the case $|\mc K|=k+1$. Induction implies that, for any $j \in \mc K$:
\be
\label{eq:S-split-eq1}
S(\mc J, \rlb, \mu) \leq \sum_{e \in \mc K \setminus \{j\}} S_e(x_e)+ S(\mc J \cup \{j\}\setminus \mc K, x, \mu).
\ee
Applying induction to the second term in \eqref{eq:S-split-eq1}, for every $z \in g(\mc J \cup \{j\} \setminus \mc K, x, \mu)$ we get
 \be
\label{eq:S-split-eq2}
S(\mc J \cup \{j\} \setminus \mc K, x, \mu) \leq S_j(z_j)+ S(\mc J \setminus \mc K, z, \mu).
\ee
$z \in g(\mc J  \cup \{j\} \setminus \mc K, x, \mu)$ implies $z \geq x$, which in turn implies $S_j(z_j) \leq S_j(x_j)$ and $S(\mc J \setminus \mc K, z, \mu) \leq S(\mc J \setminus \mc K, x, \mu)$ Lemma~\ref{lem:s-function-decreasing}. Combining this with \eqref{eq:S-split-eq1} and \eqref{eq:S-split-eq2} establishes the lemma for $|\mc K|=k+1$, and hence for all $\mc K \subseteq \mc J$ by induction.
\end{proof}

%\ksmargin{is $\mc E^v$ notation necessary ?}
For the next lemma, we again allow the possibility that $(\mc V(0),\mc E(0)) \neq (\mc V, \mc E)$. 
%The lemmas  use the following notation: for $v \in \mc V \setminus \{n\}$,
%\be
%\label{eq:Ev-def}
%\mc E^v:=\setdef{e \in \mc E}{\exists \text{ a directed path in } (\mc V, \mc E) \text{ from }v \text{ to } \sigma_e \text{ and } \tau_e}.
%\ee

\begin{lemma}
\label{lem:lb-lem1}
Consider a node $v \in \mc V \setminus \{n\}$ with $|\mc E_v^+| \leq 3$ in a flow network $\mc N$ with initial condition $(\mc V(0),\mc E(0))\subseteq (\mc V, \mc E)$, satisfying Assumption~\ref{ass:polytree}, and operating under BPA routing. 
%satisfying $G^v(\mc E_v^+(0), \mu) \succeq \rlb$ for some  
%$\rlb \in \mc X_v(\mc E_v^+(0),\zerobf,\mu) \cup \{\zerobf\}$. 
Let the inflow $\lambda_v(t)$ be non-decreasing and satisfy $\max_{t \geq 0} \lambda_v(t)=\mu \geq 0$. If $\lambda_v(t) \equiv \mu$\,, or if 
BPA routing at $v$ is flow monotone, then $v \in \mc V(\mc T)$ under any disturbance process $\delta^v$ satisfying 
%$\delta^v_e(t) \equiv 0$ for all $e \in \mc E(0) \setminus \mc E^v$ and 
$\Dnorm(\delta^v) < S(\mc E_v^+(0), \rlb^*, \mu)$, with $r^*=G^v(\mc E_v^+,\mu)$.

\end{lemma}
\begin{proof} Let $\mc E_v^+ = \{e_1,e_2,e_3\}$ and $f_i^*:=G^v_{e_i}(\mc E_v^+(0),\mu)$ for all $e_i \in \mc E_v^+(0)$. We consider three possible scenarios for $|\mc E_v^+(0)|$ separately.
We prove by backward induction on $v$ in $\zerountil{n-1}$.
First consider $v=n-1$.  
%Therefore, under flow monotonicity, $G^v_{e_i}(\mc E_v^+(t),\lambda_v(t))=f_i(t) \leq f_i^*$ for $i=1, 2, 3$, $t \geq 0$.
When $\mc E_v^+(0)=\{e_1\}$,  \eqref{eq:Se-def} implies that $S_{e_1}(\mu) = \flowmax_{1} - \mu \leq \flowmax_1 - \lambda_v(t)$. Therefore, for all $t \geq 0$, $\Dcum_{e_1}(t) \leq \Dcum_{e_1}(\mc T) \leq \Dnorm(\delta^v)<S_{e_1}(\mu)=\flowmax_{1}-\mu \leq \flowmax_{1} - \lambda_v(t)$. That is, $f_{e_1}(t)=\lambda_v(t) < \flowmax_{1} - \Dcum_{e_1}(t)=\flowmax_1(t)$ for all $t \geq 0$, and hence $v \in \mc V(\mc T)$.

When $\mc E_v^+(0)=\{e_1,e_2\}$, let $t_1:=\min \setdef{t \geq 0}{\mc E_v^+(t) \neq \mc E_v^+(0)}$. 
To avoid triviality, assume $t_1 < \infty$, and let $\mc K:=\mc E_v^+(0) \setminus \mc E_v^+(t_1)$ be the set of links to become inactive simultaneously at $t_1$. Then, necessarily $\Dcum_e(t_1) \geq [\flowmax_e-f_e(t_1)]^+ \geq [\flowmax_e - f_e^*]^+$ for all $e \in \mc K$, where the second inequality follows from flow monotonicity. Therefore,
\begin{equation}
\label{eq:counter-example-proof}
\sum_{e \in \mc K} \Dcum_e(\mc T) \geq  \sum_{e \in \mc K} \Dcum_e(t_1) \geq \sum_{e \in \mc K} [\flowmax_e-f_e^*]^+. 
\end{equation}
Let $\tdelta^{\mc K}$ be such that $\tdelta^{\mc K}_e(t) \equiv 0$ for all $e \in \mc K$ and $\tdelta^{\mc K}_e(t) \equiv \delta^v_e(t)$ for all $e \in \mc E_v^+(t_1)$. Therefore, 
\be
\label{eq:cost-split}
\Dnorm(\delta^v)=\Dnorm(\tdelta^{\mc K}) + \sum_{e \in \mc K} \Dcum_e(\mc T).
\ee 
Combining \eqref{eq:counter-example-proof} with Lemma 4, where we note that $f^*:=G^v(\mc E_v^+(0),\mu) \in g(\mc E_v^+(0),\rlb^*,\mu)$ (from \eqref{eq:BPA-routing}), we get 
\be
\label{eq:BPA-implication}
S(\mc E_v^+(0), \rlb^*, \mu) \leq \sum_{e \in \mc K} \Dcum_e(\mc T) + S(\mc E_v^+(t_1), f^*,\mu)
\ee

Combining \eqref{eq:BPA-implication} and \eqref{eq:cost-split} with $\Dnorm(\delta^v) < S(\mc E_v^+(0), \rlb^*, \mu)$, we get 
\begin{equation}
\label{eq:smallD-ub}
\Dnorm(\tdelta^{\mc K}) < S(\mc E_v^+(t_1), f^*,\mu).
\end{equation} 
If $|\mc E_v^+(t_1)|=0$, then \eqref{eq:smallD-ub} is a contradiction, and if $|\mc E_v^+(t_1)|=1$, the proof is then completed by using the $|\mc E_v^+(0)|=1$ case since $S(\mc E_v^+(t_1), f^*,\mu)=S(\mc E_v^+(t_1), \zerobf,\mu)$.

For $\mc E_v^+(0)=\{e_1,e_2, e_3\}$, one follows the same argument as before to arrive at \eqref{eq:smallD-ub}. If $|\mc E_v^+(t_1)| \leq 1$, then we use the same arguments as before. If $|\mc E_v^+(t_1)| = 2$, then necessarily $\mc E_v^+(0)=\mc E_v^+$, in which case $f^*=\rlb^*$. Therefore, one can continue with the 
$|\mc E_v^+(0)|=2$ case to complete the proof.  This proves the lemma for $v=n-1$. 

Assume that the lemma is true for all $v \geq \ell$ for some $\ell \geq 1$. Let $v=\ell-1$. 
%The proof structure is similar to that of Lemma~\ref{lem:upperbound-general-lemma} and Lemma~\ref{lem:lb-lem1}, and hence we do not give all the details for the arguments repeated in this proof.
%We prove by (backward) induction on the topological ordering of $v$ in $\zerountil{n-1}$. 
%If $v=n-1$, then it is a leaf node and the result follows from Lemma~\ref{lem:lb-lem2}.
%Assume that the lemma is true for all $v \geq \ell$ for some $\ell \geq 1$. 
%Let $\mc E_v^+=\{e_1,\ldots,e_m\}$ for some $m \geq 1$. 
%Let $f_i(t)$ be the flow on link $e_i$ under BPA routing at time $t \geq 0$. 
We recall Lemma~\ref{lem:monotonicity} for monotonicity of $f_i(t)$. As for the $v=n-1$, we consider three cases depending on the value of $|\mc E_v^+(0)|$
and provide proof using similar arguments. We provide a few details only for $\mc E_v^+(0)=\mc E_v^+=\{e_1,e_2,e_3\}$. Lemma~\ref{lem:monotonicity} implies that $f_i(t) \leq f_i^*=G_i^v(\mc E_v^+, \mu)$ for all $i =1, 2, 3$. 
Let $\delta^v =\sum_{i=0}^3 \delta^v_i$, where, for $i=1,2,3$, $\delta^v_i$ is the component of $\delta^v$ on links consisting of $e_i$ and the sub-tree rooted at $\tau_{e_i}$, and $\delta^v_0$ is the component on the rest of the links in the network. Let $\Dcum_i(t)$, $i=0,1,2,3$, be defined accordingly. 
Let $t_1:=\min \setdef{t \geq 0}{\mc E_v^+(t) \neq \mc E_v^+}$. 
To avoid triviality, assume $t_1 < \infty$, and let $\mc K:=\mc E_v^+ \setminus \mc E_v^+(t_1)$ be the set of links to become inactive simultaneously at $t_1$.
Then, necessarily $\Dcum_e(t_1) \geq S_e(f_e(t_1)) \geq S_e(f_e^*)$ for all $e \in \mc K$, where the second inequality follows from Lemma~\ref{lem:s-function-decreasing}. Following similar arguments as before, we arrive at \eqref{eq:smallD-ub}, after which we use the relevant case depending on the value of $|\mc E_v^+(t_1)| \leq 2$. This establishes the proof for 
$v=\ell-1$, and hence by backward induction for all $v \in \zerountil{n-1}$.
\end{proof}

Theorem~\ref{thm:lower-bound} is obtained from Lemma~\ref{lem:lb-lem1} by substituting $v=0$, $(\mc V(0), \mc E(0))=(\mc V, \mc E)$, $\lambda_v(t) \equiv \lambda$, and noting that $S(\mc E_0^+,\rlb^*,\lambda) = S(\mc E_0^+,\zerobf,\lambda)=S^*(\mc N, \lambda)$.

\subsection{Proof of Proposition~\ref{prop:monotonicity-sufficient-conditions-graphical}}
%We start with a useful analytical property of BPA routing when $|\mc E_v^+| \leq 2$.  
%
% \begin{lemma}
%\label{lem:optimal-routing-prop}
%For every $\xi^v \in \{0,1\}^{\mc E_v^+}$, $v \in \mc V \setminus \{n\}$ with $\|\xi^v\|_1 \leq 2$ and $\mu \in \left[0, \sum_{e \in \mc J} \BPAflowmax_e\right]$, where $\mc J \subseteq \mc E_v^+$ is such that $\onebf_{\mc J}=\xi^v$, BPA routing is a singleton and continuous in $\mu$.
%\end{lemma}
%\kscommentphantom{
%\begin{proof}
%(work in progress ...)
%\end{proof}
%}

%Part (i) of the theorem follows from the fact that the origin node is excluded from the definition of flow-monotone property, since the inflow at the origin node is fixed. Part (ii) follows from the fact that, for a symmetric flow network, BPA routing is flow (and link)-monotone. We now provide proof for part (iii).
% when $\mc N$ is one of the five sub-tree topologies illustrated in Figure~\ref{fig:graphical-sufficient-conditions} and satisfies corresponding additional conditions (if any) on the flow capacities as stated in the theorem. We then extend to general $\mc N$ satisfying the sufficient conditions in the statement of Theorem.

The BPA routing for case (a) is explicitly computed in \eqref{eq:Gstar1-caseb}, which readily implies flow-monotonicity. 
%The proof for case (b) follows trivially from case (a), and 
The proof for case (c) follows from the constructs used in the proof of case (b). 
Therefore, we provide details only for case (b).

For brevity in notation, let $y(\mu) \equiv G^{v_0}(\mc E_{v_0}^+,\mu)$
be the flow under BPA routing.  For brevity in notation, and since the lower bound constraints imposed by $\rlb$ are redundant in this case, we drop the dependence of $S(.\rlb,.)$ on $\rlb$.
Following Figure~\ref{fig:Rv0}(b), the general relationship between $S_{e_1}(\mu)$ and $S(\mc E_{v_1}^+,\mu)$ can be written as:
\begin{equation}
\label{eq:S1-cased}
S_{e_1}(\mu) =
 \left\{\ba{lcl} 
S(\mc E_{v_1}^+,\mu) & \se & \mu \in [\bar{\mu}_1,\bar{\mu}_2], \\
\left[\flowmax_{1}-\mu\right]^+ & \se & \mu \in  [0,\bar{\mu}_1] \cup [\bar{\mu}_2,\flowmax_{1}] \,, 
\ea \right. 
\end{equation}
where $\bar{\mu}_1=2\left(\flowmax_{3} +\flowmax_{4}-\flowmax_{1}\right)$, $\bar{\mu}_2=2 \flowmax_{1} - (\flowmax_{3} + \flowmax_{4})$ and one can write an expression for $S(\mc E_{v_1}^+,\mu)$ similar to \eqref{eq:Rv-caseb}. We prove flow-monotonicity by showing that $\frac{\de}{\de \mu}y(\mu) \geq 0$. 
Let $\BPAflowmax_i=\sup\{\mu: \, S_{e_i}(\mu) > 0\}$ for $i=1,2$ be the effective flow capacity of link $e_i$. 
When $y(\mu)$ is on the boundary of the feasible set $\mc X_{v_0}(\mc E_{v_0}^+,\mu)$, without loss of generality, assume that $y_1(\mu)=[\mu-\BPAflowmax_2]^+$ and $y_2(\mu)=\min\{\mu, \BPAflowmax_2\}$, which is trivially flow-monotone. 
When $y(\mu)$ is in the interior of the feasible set, $y_2$ satisfies $S_{e_2}(y_2) + S_{e_1}(\mu)=S_{e_1}(\mu-y_2)+S_{e_2}(\mu)$. Therefore, by the implicit function theorem,\footnote{$S_{e_1}(\mu)$ and $S_{e_2}(\mu)$ are continuous piecewise linear functions, e.g., see Figure~\ref{fig:Rv0} (c). Therefore, one can find analytic functions that can approximate $S_{e_1}(\mu)$ and $S_{e_2}(\mu)$, as well as their derivative arbitrarily closely at almost all points. We implicitly assume that $S_{e_1}(\mu)$ and $S_{e_2}(\mu)$ are replaced with such analytic approximations.} we have that 
\be
\label{eq:x3star-derivative-twolinks}
 \frac{d}{d\mu}y_2(\mu)=\frac{S'_{e_2}(\mu)+S'_{e_1}(\mu-y_2)-S'_{e_1}(\mu)}{S'_{e_2}(y_2)+S'_{e_1}(\mu-y_2)},
\ee 
where $S'_{e_i}(y) \equiv \frac{\de}{\de y} S_{e_i}(y)$, $i=1,2$. The strictly decreasing property of $S_{e_i}$, $i=1,2$ from Lemma~\ref{lem:s-function-decreasing} implies that the denominator of \eqref{eq:x3star-derivative-twolinks} is negative for all $\mu \in \left(0, \BPAflowmax_{1} + \BPAflowmax_{2} \right)$. \eqref{eq:Rv-caseb} and \eqref{eq:S1-cased} imply that $S'_{e_1}(\mu-y_2)-S'_{e_1}(\mu) \leq 1/2$ if $\mu < \flowmax_1$ and 
equal to zero if $\mu > \flowmax_1$. This combined with $S_{e_2}(\mu) \equiv [\flowmax_2-\mu]^+$ and the assumption that $\flowmax_2 \geq \flowmax_1$ implies that the numerator of \eqref{eq:x3star-derivative-twolinks} is non positive for all $\mu < \BPAflowmax_1 + \BPAflowmax_2$. Hence,  
%One can get expression for $R_1(\mu)$ from case (b). 
 $\frac{d}{d\mu}y_2(\mu) \geq 0$.
The proof for $ \frac{d}{d\mu}y_1(\mu) \geq 0$ follows along similar lines. 

\section{Conclusions and Future Work}
\label{sec:conclusions}
In this paper, we proposed a dynamical model for cascading failures in single-commodity network flows, where the network dynamics is governed by a deterministic and possibly adversarial disturbance process which incrementally reduces flow capacity on the links, and distributed oblivious routing policies that have information only about the 
local inflow and active status of outgoing links, and in particular no information about the disturbance process. The salient feature of this
model is to couple the flow dynamics with the link and node inactivation dynamics. An immediate outcome of this coupling is that, links and nodes to fail successively are not necessarily adjacent to each other. 
We quantified margin of resilience to be the minimum cumulative capacity reductions across time and links of the network, under which the network looses its transferring property. We presented an algorithm that provides an upper bound on the margin of resilience for directed acyclic graphs between a single origin-destination pair. The same algorithm motivates a routing policy which provably matches the upper bound for networks which are tree like, have out-degree at most 3, and induce monotonicity in the flow dynamics. 

In future, we plan to extend our analysis to networks with general graph topologies, multi-commodity flows, non-oblivious routing policies with possibly multi hop information, 
stochastic disturbance processes, reversible link activation dynamics under finite time link recovery, and exogenous coupling between failure and recovery of distant links due to coupling between the given network and other exogeneous networks. We also plan to investigate computationally efficient, and possibly distributed, algorithms for (approximate) computation of the margin of resilience. Finally, we plan to extend our formulation to the physics of other networks such as power, gas, water and supply chains.

\bibliographystyle{IEEEtran}
\bibliography{ksmain,savla} 

\appendix 
In this section, we state and prove certain useful properties of $S(\mc  J, \rlb,\mu)$ defined in \eqref{eq:S-tilde-def-new}.

\begin{lemma}
\label{lem:s-function-decreasing}
Consider two networks $\mc N^1=(\mc V, \mc E^1, \flowmax^1)$ and $\mc N^2=(\mc V, \mc E^2, \flowmax^2)$, each satisfying Assumption~\ref{ass:acyclicity}. Let $S^1$ and $S^2$ be the functions computed by the Backward Propagation Algorithm for $\mc N^1$ and $\mc N^2$ respectively. For any $v \in \mc V \setminus \{n\}$, let $\mc J$ be any subset of links common to $\mc N^1$ and $\mc N^2$, and outgoing from $v$ .
Then,
$$0 \leq \mu_2 \leq \mu_1, \, \, \zerobf \leq \rlb^2 \leq \rlb^1, \, \,  \mc E^1 \subseteq \mc E^2 \text{ and } \flowmax^1 \leq \flowmax^2 \implies S^1(\mc K, \rlb^1, \mu_1) \leq S^2(\mc J, \rlb^2, \mu_2),$$
where $\mc K:=\setdef{e \in \mc J}{S^1_e(\rlb^1_e) > 0}$.
\end{lemma}
\begin{proof}
We split the lemma as follows:
\begin{enumerate}
\item $S^1(\mc K, \rlb^1,\mu_1) = S^1(\mc J, \rlb^1,\mu_1)$;
\item $S^1(\mc J, \rlb^1, \mu_1) \leq S^2(\mc J,\rlb^1,\mu_1)$;
\item $S^2(\mc J, \rlb^1,\mu_1) \leq S^2(\mc J, \rlb^1,\mu_2)$; and
\item $S^2(\mc J, \rlb^1,\mu_2) \leq S^2(\mc J, \rlb^2,\mu_2)$.
\end{enumerate}
Out of these, (iv) is trivial, and hence we omit its proof. We prove (i)-(iii) by double induction, on the number of nodes $n+1$ and the cardinality of $\mc J$. 
It is immediate to verify that the claim holds true for $n=1$ and $|\mc J|=1$, since with $\mc J=\{e\}$, $S^i(e,\rlb,\mu)=S^i_{e_1}(\mu)=[\flowmax^i_{e_1} - \mu]^+$ for all $\rlb$, $i=1,2$. Assume that the claim is true for $|\mc V|=n+1$ and for $|\mc J|\le k$. Consider a $\mc J$ of cardinality $k+1$. \eqref{eq:S-tilde-def-new} and \eqref{eq:g-def-1} imply that, for all $z \in g(\mc K \cup \{j\},\rlb^1,\mu_1)$, $j \in \mc J \setminus \mc K$:
\be
\label{eq:S1-difficult-proof-eq1}
S^1(\mc K \cup \{j\},\rlb^1,\mu_1) \leq S^1_j(z_j) + S^1(\mc K, z, \mu_1) \leq S^1(\mc K, \rlb^1, \mu_1),
\ee
where the second inequality follows from the fact that $z \geq \rlb^1$ implies $0 \leq S^1_j(z_j) \leq S^1_j(\rlb^1_j)  = 0$, and using (iv) from induction. On the other hand, consider $y$ such that $y_e=\rlb^1_e$ for all $e \in \mc K$ and $y_j = \mu_1 - \sum_{e \in \mc K} \rlb^1_e \geq \rlb^1_j$, where the inequality follows from the feasibility of $y$. For such a $y$, \eqref{eq:S-tilde-def-new} implies that there exists $i \in \mc K \cup \{j\}$ such that 
\be
\label{eq:S1-difficult-proof-eq2}
S^1(\mc K \cup \{j\},\rlb^1,\mu_1) \geq S^1_i(y_i) + S^1(\mc K \cup \{j\} \setminus \{i\}, y, \mu_1)
\ee
Consider \eqref{eq:S1-difficult-proof-eq2} under two cases. (a) $i = j$. In this case, since $S^1_j(y_j) = 0$, \eqref{eq:S1-difficult-proof-eq2} gives $S^1(\mc K \cup \{j\},\rlb^1,\mu_1) \geq  S^1(\mc K \cup \{j\} \setminus \{j\}, y, \mu_1)=S^1(\mc K, \rlb^1, \mu_1)$, where the last equality follows from the fact that the components of $\rlb^1$ and $y$ along $\mc K$ are the same. (b) $i \neq j$. In this case, recalling that $y_e=\rlb^1_e$ for all $e \in \mc K$, and applying (i) from induction to the second term in \eqref{eq:S1-difficult-proof-eq2}, we get that 
\be
\label{eq:S1-difficult-proof-eq3}
S^1(\mc K \cup \{j\},\rlb^1,\mu_1) \geq S^1_i(\rlb^1_i) + S^1(\mc K \setminus \{i\}, y, \mu_1)=S^1_i(\rlb^1_i) + S^1(\mc K \setminus \{i\}, \rlb^1, \mu_1).
\ee 
For every $z \in g(\mc K, \rlb^1,\mu_1)$, we have
\be
\label{eq:S1-difficult-proof-eq4}
S^1_i(\rlb^1_i) + S^1(\mc K \setminus \{i\}, \rlb^1, \mu_1) \geq S^1_i(z_i) + S^1(\mc K \setminus \{i\}, z, \mu_1) \geq S^1(\mc K, \rlb^1, \mu_1),
\ee
where the first inequality follows from $z \geq \rlb^1$ and (iii) from induction.
Combining \eqref{eq:S1-difficult-proof-eq3} and \eqref{eq:S1-difficult-proof-eq4}, we arrive at the same conclusion as case (a), i.e., 
\be
\label{eq:S1-difficult-proof-eq5}
S^1(\mc K \cup \{j\},\rlb^1,\mu_1) \geq S^1(\mc K, \rlb^1, \mu_1).
\ee
Combining \eqref{eq:S1-difficult-proof-eq1} and \eqref{eq:S1-difficult-proof-eq5}, we establish (ii) when $|\mc J \setminus \mc K|=1$. The proof for arbitrary $|\mc J \setminus \mc K|$ follows from repetitive application of this procedure. 

The proof for (ii) easily follows from induction since, for every $x$, and $e \in \mc J$,
$$S^1_e(x_e) + S^1(\mc J \setminus \{e\}, x, \mu_1) \leq S^2_e(x_e) + S^2(\mc J \setminus \{e\}, x, \mu_1).$$

In order to prove (iii), using the fact that $S_e^2(\mu)$ is non-increasing in $\mu$ and $S^2(\mc J\setminus\{e\},\rlb,\mu)$ non-increasing in $\rlb$ from induction, 
one gets that, for all $\rlb^1$ such that $\onebf'\rlb^1 \le\mu_2$, 
%\kscommentphantom{(the equivalence between constrained and unconstrained optimization needs more explanation)}
$$\max_{\substack{x\ge\rlb^1\\ \onebf'x\ge\mu_2}}\min_{e\in\mc J}\l\{S^2_e(x_e)+S^2(\mc J\setminus\{e\},x,\mu_2)\r\}
=\max_{\substack{x\ge\rlb^1\\ \onebf'x=\mu_2}}\min_{e\in\mc J}\l\{S^2_e(x_e)+S^2(\mc J\setminus\{e\},x,\mu_2)\r\}
=S^2(\mc J,\rlb^1,\mu_2)\,.
$$
Specifically, to see why the first equality above holds true, let the maximum in the rightmost side be achieved in some $x\ge\rlb^1$ such that $\onebf'x\ge\mu_2$, and let $\lambda\in[0,1]$ be such that $(1-\lambda)\onebf'x+\lambda\onebf'\rlb^1=\mu_2$ (such $\lambda$ exists since $\onebf'x\ge\mu_2$ and $\onebf'\rlb^1 \le\mu_2$). Then, $y:=(1-\lambda)x+\lambda\rlb^1$ satisfies $x\ge y\ge\rlb^1$, $\onebf'y=\mu_2$, and $S^2_e(y_e)+S^2(\mc J\setminus\{e\},y,\mu_2)\ge S^2_e(x_e)+S^2(\mc J\setminus\{e\},x,\mu_2)$ for all $e$. 
Then, for $\mu_1 \ge \mu_2$, 
\begin{multline*}
S^2(\mc J,\rlb^1,\mu_1)
=\max_{\substack{x\ge\rlb^1\\ \onebf'x=\mu_1}}\min_{e\in\mc J}\l\{S^2_e(x_e)+S^2(\mc J\setminus\{e\},x,\mu_1)\r\} \\
\le\max_{\substack{x\ge\rlb^1\\ \onebf'x\ge\mu_2}}\min_{e\in\mc J}\l\{S^2_e(x_e)+S^2(\mc J\setminus\{e\},x,\mu_2)\r\}
=S^2(\mc J,\rlb^1,\mu_2)\,.
\end{multline*}

This concludes the proof for $|\mc V|=n+1$ and $|\mc J| \leq k+1$. A similar argument allows one to extend the validity of the result to $|\mc V|=n+2$ with $|\mc J|=1$. The lemma then follows by induction. 
\end{proof}

%\appendix
%\kscomment{
%We briefly discuss a procedure to modify \eqref{eq:BPA-routing}
%to guarantee that BPA routing satisfies Lemma~\ref{lem:optimal-routing-prop}
%even when $\|\xi^v\|_1> 2$.
%}

\end{document}